\documentclass[letterpaper, 10pt, conference]{ieeeconf}

\let\labelindent\relax

\usepackage{amsmath,amssymb,paralist,subfigure,cite,graphicx,color,amsbsy,float}
\usepackage{stmaryrd,mathrsfs,url}
\usepackage[table]{xcolor}
\usepackage{cuted,mathtools,lipsum}
\usepackage[ruled,vlined,linesnumbered]{algorithm2e}

\usepackage{pifont}
\usepackage[utf8]{inputenc} 
\usepackage{hyperref}       
\hypersetup{
    colorlinks=true,
    linkcolor=cyan,
    filecolor=mnodea,      
    urlcolor=cyan,
    citecolor=lime,
}
\usepackage{url}            
\usepackage{booktabs}       
\usepackage{amsfonts,amsmath,amssymb}       
\usepackage{amsthm}     
\usepackage{nicefrac}       
\usepackage{microtype}      
\usepackage{lipsum}
\usepackage{graphicx}
\usepackage{dsfont}
\usepackage{enumitem}

\usepackage{baskervillef}
\usepackage{caption}
\usepackage{pgfplots}
\pgfplotsset{compat=newest}
\usetikzlibrary{plotmarks}
\usetikzlibrary{arrows.meta}
\usepgfplotslibrary{colorbrewer}
\usepackage{multicol}
\usepackage{breqn}

\usepackage{amssymb,amsmath,amsfonts}
\usepackage{algorithmic}

\usetikzlibrary{automata,chains,arrows}
\usetikzlibrary{arrows.meta}
\usepackage{flushend}

\floatstyle{ruled}
\newfloat{model}{H}{mod}
\floatname{model}{\footnotesize Model}
\newfloat{notatio}{H}{not}
\floatname{notatio}{\footnotesize Notation}

\newtheorem{remark}{\bfseries Remark}
\newtheorem{theorem}{\bfseries Theorem}
\newtheorem{lemma}{\bfseries Lemma}
\newtheorem{assumption}{\bfseries Assumption}
\newtheorem{corollary}{\bfseries Corollary}

\let\mathbb=\mathds 

\def\mb{\mathbf}
\def\mc{\mathcal}
\def\Y{\mathcal{Y}}

\floatstyle{ruled}
\newfloat{model}{H}{mod}
\floatname{model}{\footnotesize Model}
\newfloat{notatio}{H}{not}
\floatname{notatio}{\footnotesize Notation}

\definecolor{themisblue}{rgb}{0.76, 0.89, 1.0}

\IEEEoverridecommandlockouts                             
\begin{document}
%
%
%
%
\title{ \bf \textsf{DTAC-ADMM}: Delay-Tolerant Augmented Consensus ADMM-based Algorithm for Distributed Resource Allocation}

\author{Mohammadreza  Doostmohammadian, Wei Jiang, and Themistoklis Charalambous
\thanks{The authors are with the School of Electrical Engineering at Aalto University, Finland (\texttt{name.surname@aalto.fi}).
M.  Doostmohammadian is also with the Faculty of Mechanical Engineering, Semnan University, Iran (\texttt{doost@semnan.ac.ir}). 
T. Charalambous is also with the Electrical and Computer Engineering Department, University of Cyprus, Cyprus, \texttt{surname.name@ucy.ac.cy}
}}
\maketitle
\thispagestyle{empty}

\maketitle
\thispagestyle{empty}
\pagestyle{empty}

%
%
%
%
\begin{abstract}
Latency is inherent in almost all real-world networked applications. In this paper, we propose a distributed allocation strategy over multi-agent networks with delayed communications. The state of each agent (or node) represents its share of assigned resources out of a fixed amount (equal to overall demand). Every node locally updates its state toward optimizing a global allocation cost function via received information of its neighbouring nodes even when the data exchange over the network is heterogeneously delayed at different links. The update is based on the alternating direction method of multipliers (ADMM) formulation subject to both sum-preserving coupling-constraint and local box-constraints. The solution is derivative-free and holds for general (not necessarily differentiable) convex cost models. We use the notion of augmented consensus over undirected networks to model delayed information exchange and for convergence analysis. We simulate our \textit{delay-tolerant} algorithm for optimal energy reservation-production scheduling. 
\end{abstract}


%
%
%
%
\section{Introduction}
\label{sec_intro}
Resource allocation refers to the optimal assignment of the (fixed) amount of resources minimizing certain resource-related costs. 
Distributed algorithms to solve resource allocation problems over a network have recently emerged in literature, primarily motivated by advances in distributed data processing and cloud computing (instead of centralized single-node solutions) to handle more complex optimization problems and large-scale data sets. Applications include distributed energy resource provisioning for ancillary services \cite{Themis_allerton}, automatic generation control over the power grid \cite{mrd20211st}, and CPU scheduling over networked data centres \cite{rikos2021optimal}.  

Existing distributed solutions span from primary linear Laplacian-gradient \cite{shames2011accelerated,boyd2006optimal, cherukuri2015distributed} strategies with all-time feasibility to more recent primal-dual solutions \cite{aybat2016distributed,Nedic2018,falsone2020tracking} reaching feasibility asymptotically over time. The Laplacian-gradient solutions are known to preserve anytime-feasibility even in the presence of model nonlinearities (over uniformly-connected networks) \cite{fast,mrd20211st}, where the results can further address quantization/saturation nonlinearities \cite{rikos2021optimal,fast,wu2020distributed} and finite/fixed-time convergence \cite{garg_cdc,taes2020finite}. The feasibility constraint defines the balance between the assigned resources and the overall demand. This global constraint couples the states of all computing nodes and thus is referred to as the \textit{coupling-constraint} \cite{falsone2020tracking}. For example, in the generator coordination and economic dispatch problem (EDP), the generated power (assigned resources) needs to meet the grid's power demand. Otherwise, violating the solution feasibility may lead to disruption in the service and even system breakdown  \cite{wu2021new}. In this sense, with feasible initialization, the outcome of these algorithms is feasible at any termination time. 
On the other hand, primal-dual solutions need to reach feasibility (sufficiently fast) with no need for feasible initialization.
In this scenario, \textit{alternating direction method of multipliers (ADMM)} can be adopted to divide the main optimization problem into sub-problems to be locally optimized at each node \cite{banjac2019decentralized,wu2021new,falsone2020tracking,carli_admm,chang_admm,arxiv_digraph,Teixeira_admm,Wei_cdc}. In general, ADMM-based solutions, being gradient-free, make less restrictive assumptions on the smoothness of the objective function. Some ADMM formulations, further, allow for handling the local constraint on the state (the \textit{box-constraint}) at each node, e.g., via first-order proximal approximation of the cost for min-max problem \cite{chang_admm}. Other works consider penalty (or barrier) functions to address this constraint \cite{cherukuri2015distributed,wu2021new,fast,mrd20211st}.

In this work, reformulating the primal distributed resource allocation (DRA) problem to its dual form and assuming standard strong-duality, we use the ADMM-based method inspired by \cite{falsone2020tracking} to find the optimal solution in a distributed heterogeneous way. As the main contribution, we propose a \textit{delay-tolerant} algorithm in case the communications among different nodes are subject to different time delays. These heterogeneous delays (at different links) imply heterogeneous communications over the network. In this sense, this work improves the recent primal-dual algorithms over non-delayed networks \cite{Wei_cdc,banjac2019decentralized,wu2021new,falsone2020tracking,carli_admm,chang_admm} to more advanced delay-tolerant resource allocation solutions. Combining the notions of augmented consensus and ADMM, 
the proposed \textsf{DTAC-ADMM} algorithm is single time-scale, i.e., it performs \textit{only one step of (augmented) consensus on (delayed) information per ADMM iteration}. This is computationally more efficient than double time-scale solutions with many steps of consensus (as inner-loop) per iteration \cite{arxiv_digraph}. Further, \textsf{DTAC-ADMM} is initialization-free (with no need for primal feasibility), i.e., it does not require the initial states to satisfy the feasibility constraint but gains feasibility over time. This improves \cite{shames2011accelerated,boyd2006optimal,cherukuri2015distributed} which consider specific initialization in order to properly run their main allocation in the upcoming steps. Similar to the (continuous-time) dynamics in \cite{wang2019distributed}, \textsf{DTAC-ADMM} can handle large delays, but without need for specific switching techniques which may cause feasibility-gap. 
We use \textsf{DTAC-ADMM} for scheduling over the energy grid to optimally allocate the share of energy reservation and generation.

\textbf{Organization:} Section \ref{sec_prob} formulates the DRA problem (both primal and dual formulations). Section~\ref{sec_sync} and \ref{sec_async} respectively present the ADMM solution under homogeneous and heterogeneous delays. Applications and simulations are given in Sections~\ref{sec_app} and~\ref{sec_sim}. Section~\ref{sec_con} concludes the paper.  

\textbf{Notations:} 
We use small letters for scalars, small bold letters for vectors, and underlined small bold letters for augmented vectors. Capital letters denote matrices. Superscript $k$ denotes the discrete time index. $\|\cdot\|$, ``;'', $\otimes$ denote the $2$-norm, column concatenation, and Kronecker matrix product. $\mb{1}_m$ ($\mb{0}_m$) and $\mb{1}_{m\times m}$ ($\mb{0}_{m\times m}$) denote all ones (and zeros) column vector and matrix of size $m$. $\mb{I}_{m}$ denotes the size $m$ identity matrix. PD (PSD) stands for positive (semi-) definite. 

\section{Problem formulation} \label{sec_prob}
Consider a group of $n$ agents, each associated with a local cost $\phi_i(y_i):\mathbb{R} \rightarrow \mathbb{R}$\footnote{The results can be generalized to $\mathbb{R}^p,p>1$, and considering $p=1$ is only for convenience in presenting the results.}. State $y_i$ represents the allocated resources to agent $i$ where the overall resources are constant for the group of agents. 
Then, DRA is formulated as,
\vspace{-0.2cm}
\begin{equation}\label{problem_initial_0}
	\begin{aligned}
		\operatorname*{min}_{\mb{y} = [y_1;\dots;y_n]  \in \mathbb{R}^{n} }   & \Phi(\mb{y})  =  \operatorname*{min}_{y_i\in \mathbb{R} }  \sum_{i=1}^{n} \phi_i(y_i), \\
		\mbox{subject~to}&\sum_{i=1}^n ({y}_i - b_i)=0,	~ y_i \in \Y_i
	\end{aligned}
\end{equation}
with $b_i\in \mathbb{R}$ and $\Y_i$ as a closed convex set (the local box-constraint) at agent $i$. These box-constraints in the form $m_i \leq y_i \leq M_i$,
limit the amount of resources at each agent. One can address these constraints either \emph{(i)}  by adding penalty/barrier functions \cite{fast,mrd20211st,wu2020distributed} or \emph{(ii)} by directly applying it in the ADMM formulation \cite{falsone2020tracking}. Note that problem~\eqref{problem_initial_0} differs from equality-constraint problems \cite{Wei_cdc,Teixeira_admm,arxiv_digraph} due to its sum-preserving constraint.

\begin{assumption}\label{assup_convex}
	Each cost function $\phi_i : \mathbb{R} \rightarrow \mathbb{R} \cup \{+\infty\}$ is closed, proper and convex (not necessarily differentiable).
\end{assumption}
The cost is not restricted to quadratic form as in  \cite{rikos2021optimal,Teixeira_admm}.
The Lagrangian form after incorporating the feasibility constraint is
\begin{equation}\label{lagrang_initial}
L(x,y) = \sum_{i=1}^{n} \phi_i(y_i) + x \sum_{i=1}^n ({y}_i - b_i),
\end{equation}
which gives the dual formulation of \eqref{problem_initial_0} as
\begin{equation} \label{problem_initial}
	\begin{aligned} 
		\max_{x\in \mathbb{R}} \min_{y_i\in \Y_i} L(x,y) &= \max_{x\in \mathbb{R}} F(x) = \max_{x\in \mathbb{R}} \sum_{i=1}^{n} f_i(x)
		\\
		f_i(x) = &\min_{y_i\in \Y_i}  (\phi_i(y_i) +x(y_i-b_i)).
	\end{aligned}
\end{equation}
In this new setup, $ x \in \mathbb{R} $ (as the Lagrange parameter) is a global optimization variable (a common decision variable) and $f_i$ is known to agent $i$ only. One can prove strong duality between the two formulations \eqref{problem_initial_0} and \eqref{problem_initial}, which in turn implies \textit{zero duality gap} between the two problems. This directly follows from \textit{Slater's theorem} giving sufficient condition, saying that strong duality holds if: \emph{(i)}  the primal problem is convex; and \emph{(ii)}  the strict feasibility holds \cite{boyd2004convex}. These are the assumption that we discuss next.


\begin{assumption}[Strict feasibility] \label{ass_feasible} 
	There exists feasible point $\mb{y}$ in the primal formulation~\eqref{problem_initial_0} that satisfies the coupling constraint $\sum_{i=1}^n (y_i - b_i)=0$ and  \textit{strictly} satisfies $y_i \in \Y_i$. The latter implies that, for the scalar $y_i$, the inequality $m_i < y_i < M_i$  strictly holds.
\end{assumption}

\begin{assumption}[Stong Duality] \label{ass_dual}
	There is \textit{zero duality gap} between primal and dual problems \eqref{problem_initial_0} and \eqref{problem_initial} as a follow-up to Assumptions~\ref{assup_convex} and \ref{ass_feasible} \cite[Section~5.2.3]{boyd2004convex}.
\end{assumption}

Note that Assumptions~\ref{ass_feasible} and \ref{ass_dual} are standard in ADMM literature \cite{banjac2019decentralized}.
The idea in this paper is to apply the ADMM structure to solve the dual problem  \eqref{problem_initial}, i.e., to find its optimal solution $\mb{y}^*=[y_1^*;y_2^*;\dots;y_n^*]$, under the given assumptions on the convexity and feasibility, over a network and in a distributed way. We consider a connected undirected network $\mc{G}$ with symmetric and bi-stochastic link weights. The matrix $W$ denotes the PSD bi-stochastic symmetric adjacency matrix (for consensus). To solve the DRA problem cooperatively, agents share their local data over $\mc{G}$ with their neighbours $\mc{N}_i$ and receive information from the same neighbours $\mc{N}_i$. This information sharing is \textit{local}  and possibly \textit{delayed}, with no global knowledge (of the overall cost $\Phi$) involved. 
\begin{assumption} \label{ass_net}
	The network $\mc{G}$ is undirected connected.
\end{assumption}
A simple choice of PSD adjacency matrix $W$ satisfying Assumption~\ref{ass_net} is to set the link weights equal to $\frac{1}{|\mc{N}_i|+1}$ \cite{Themis_delay} and, then, put $W = \frac{1}{2}(W+\mb{I}_n)$ \cite{falsone2020tracking}.  

Next, we present two existing ADMM solutions to solve problem \eqref{problem_initial_0} (and \eqref{problem_initial}) in presence of no delay over the network. 

\textsf{Parallel-ADMM}: First, recall on a semi-centralized algorithm proposed by \cite{bertsekas1989parallel}. Assume node $i$ updates its state as 
\begin{align}  \label{eq_paral_y}
    y_i^{k+1} = \operatorname*{argmin}_{y_i\in \Y_i} \left(\phi_i(y_i) + x_i y_i  + \frac{c}{2} \|y_i - y_i^k + d^k \|^2 \right)
\end{align}
where the last term penalizes the non-feasibility by factor $c>0$ via the following centralized variables
\begin{align} \label{eq_paral_d}
   d^{k+1} & = \frac{1}{n} \sum_{i=1}^n (y_i^{k+1}-b_i), \\ \label{eq_paral_x}
   x^{k+1} & = x^k + c d_{k+1}.
\end{align}

\textsf{Distributed-Parallel-ADMM}: Adopting the standard dynamic average consensus techniques, a distributed version of \eqref{eq_paral_y}-\eqref{eq_paral_x} is proposed in \cite{falsone2020tracking} as
\begin{align} 
\nonumber
    y_i^{k+1} &= \operatorname*{argmin}_{y_i\in \Y_i} \Big(\phi_i(y_i) + \sum_{j=1}^n W_{ij} x_j^k y_i  \\   \label{eq_nodelay_y}
    &+ \frac{c}{2} (y_i - y_i^k + \sum_{j=1}^n W_{ij} d_j^{k})^2 \Big), \\ \label{eq_nodelay_d}
   d_i^{k+1} & = \sum_{j=1}^n W_{ij} d_j^{k} + (y_i^{k+1}-b_i) - (y_i^{k}-b_i), \\   \label{eq_nodelay_x}
   x_i^{k+1} & = \sum_{j=1}^n W_{ij} x_j^k + c d_i^{k+1},
\end{align}
by introducing local variables $x_i^k$ and $d_i^k$ instead of the centralized variables $x^k$ and $d^k$ in \eqref{eq_paral_x} and \eqref{eq_paral_d}, respectively.

\section{ADMM Protocol for Homogeneous Delays} \label{sec_sync}
First, we provide a synchronous update over the network to address network latency. This enlightens the analysis and proofs for the main algorithm in Section~\ref{sec_async} considering heterogeneous delays.
Assuming a homogeneous delay $\overline{\tau}$ at all links, the state-update follows as
\begin{align} \nonumber
    y_i^{k+1} &= \operatorname*{argmin}_{y_i\in \Y_i} \Big(\phi_i(y_i) + \sum_{j=1}^n W_{ij} x_j^{k-\overline{\tau}}  y_i  \\ \label{eq_y_sync}
    &+ \frac{c}{2} \Big(y_i - y_i^k + \sum_{j=1}^n  W_{ij} d_j^{k-\overline{\tau}} \Big)^2 \Big), \\ \label{eq_d_sync}
   d_i^{k+1} & = \sum_{j=1}^n  W_{ij} d_j^{k-\overline{\tau}} + (y_i^{k+1}-b_i) - (y_i^{k}-b_i), \\ \label{eq_x_sync}
   x_i^{k+1} & = \sum_{j=1}^n W_{ij} x_j^{k-\overline{\tau}} + c d_i^{k+1}.
\end{align}
We simplify this formulation in compact vector form as   

\small   \begin{align} \label{eq_y_sync_vect}
    \mb{y}^{k+1} & = \operatorname*{argmin}_{\mb{y} \in \Y} \Big(\Phi(\mb{y}) + \mb{y}^\top W \mb{x}^{k-\overline{\tau}} + \frac{c}{2} \|\mb{y} - \mb{y}^k + W \mb{d}^{k-\overline{\tau}} \|^2 \Big),\\ \label{eq_d_sync_aug}
	\mb{d}^{k+1} &= W \mb{d}^{k-\overline{\tau}} +  \mb{y}^{k+1}-\mb{y}^{k} , \\ \label{eq_x_sync_aug}
	\mb{x}^{k+1} &= W \mb{x}^{k} + c \mb{d}^{k-\overline{\tau}},
\end{align}    \normalsize
with $\mb{x}^{k},\mb{d}^{k+1}$ as column concatenation of $x_i^k,d_i^k$. 
This scenario is equivalent to updating at a longer iteration (scaled by the delay-bound $\overline{\tau}$) than the current iteration  $k$, see \cite[Remark~3]{Themis_delay}.
.
\section{Augmented Consensus ADMM: Heterogeneous Delays } \label{sec_async}
In this section, we move one step further and assume that the data exchange is subject to generally \textit{heterogeneous} networking delays satisfying the following assumption.
\begin{assumption} \label{ass_delay}
	The delays are assumed heterogeneous at different links (but symmetric over the same undirected link), arbitrary, bounded, and time-invariant.
\end{assumption} 

The time delay over (both sides of) the link between node $i$ and $j$ is assumed an integer value satisfying $0\leq \tau_{ij} = \overline{\tau}_{ji} \leq \overline{\tau}$. The bound $\overline{\tau}$ is only to ensure no packet loss over the network as discussed in \cite{Themis_allerton,Themis_delay}.
Assuming generally non-equal communication delays at different links implies the heterogeneous scenario in this work.  
Due to this heterogeneity, for two nodes $j,l \in \mc{N}_i$, time delay $\tau_{ij}$ differs from $\tau_{il}$, in general.
In this sense, Assumption~\ref{ass_delay} is not restrictive over non-switching networks with fixed nodes, e.g., static sensor networks \cite{delay_est}. This is well-justified over fading channels as in information-theoretic perspective assuming when the data leaves the buffer reaches the destination with a fixed delay \cite{995554}. First, define the indicator function $\mc{I}_{k,ij}(r)$ as \cite{Themis_delay}  
\begin{equation}
	\mc{I}_{k,ij}(r) = \left\{
	\begin{array}{ll}
		1, & \text{if}~ \tau_{ij}[k]=r  \\
		0, & \text{otherwise}.
	\end{array}\right.
\end{equation}
Following Assumption~\ref{ass_delay},
$\sum_{r=0}^{\overline{\tau}} \mc{I}_{k-r,ij}(r) = 1$ for two neighbors $i,j$. 
Using this definition, we propose the following heterogeneous ADMM tracking protocol to solve the DRA problem, referred to as the \textsf{DTAC-ADMM} algorithm throughout the paper. Every node $i$ updates its state according to the following discrete dynamics
\small \begin{align} \nonumber
    y_i^{k+1} &= \operatorname*{argmin}_{y_i\in \Y_i} \Big(\phi_i(y_i) + \sum_{j=1}^n \sum_{r=0}^{\overline{\tau}} W_{ij} x_j^{k-r} \mc{I}_{k-r,ij}(r) y_i  \\ \label{eq_y}
    &+ \frac{c}{2} \Big(y_i - y_i^k + \sum_{j=1}^n \sum_{r=0}^{\overline{\tau}} W_{ij} d_j^{k-r} \mc{I}_{k-r,ij}(r)\Big)^2 \Big), 
    \\ \label{eq_d}
   d_i^{k+1} & = \sum_{j=1}^n \sum_{r=0}^{\overline{\tau}} W_{ij} d_j^{k-r} \mc{I}_{k-r,ij}(r) + (y_i^{k+1}-y_i^{k}), \\ \label{eq_x}
   x_i^{k+1} & = \sum_{j=1}^n \sum_{r=0}^{\overline{\tau}} W_{ij} x_j^{k-r} \mc{I}_{k-r,ij}(r) + c d_i^{k+1},
\end{align} \normalsize
with the \textit{local initialization} at every node $i$ as
\begin{align} \label{eq_ini}
     y_i^0 \in [m_i~M_i],~
    d_i^{0} = y_i^0 - b_i, ~
    x_i^{0} = 0, 
\end{align}
where $d_i^{0}$ denotes the initial  \textit{feasibility-deviation} at node $i$, and in general is non-zero. This implies that our solution, summarized in Algorithm~\ref{alg_ac}, is initialization-free. This is an improvement over Laplacian-gradient methods, e.g., \cite{cherukuri2015distributed,boyd2006optimal}.
\begin{algorithm}[h!] 
	\caption{\textsf{DTAC-ADMM}} 
	\begin{algorithmic}[0] 
		\STATE \textbf{Input:}  $W_{ij}$, $\mc{N}_i$, $\phi_i(\cdot)$, $b_i$, $c$ \\
		\STATE   \textbf{Initialization:} Every node $i$ sets $m_i \leq y_i^0 \leq M_i$ randomly, $d_i^0 = y_i^0 - b_i$, $x_i^0 = 0$, and $k=0$ \\
        \While {termination criteria NOT true}
		\STATE Each node $i$ receives $y_j^{k-\tau_{ij}},d_j^{k-\tau_{ij}}$ from $j \in \mc{N}_i$
        \STATE Calculates Eqs.~\eqref{eq_y}-\eqref{eq_x}
        \STATE Shares $y_i^{k+1},d_i^{k+1}$ with neighboring nodes $j \in \mc{N}_i$
        \STATE Sets $k \leftarrow k+1$
   	    \STATE \textbf{Output:}  $\mb{y}^{k+1}$, $\mb{d}^{k+1}$, $\Phi(\mb{y}^{k+1}) = \sum_{i=1}^{n} \phi_i(y^{k+1}_i)$
 	    \end{algorithmic}  
   	    \label{alg_ac}
\end{algorithm} \vspace{-0.5cm}


\subsection{Augmented Consensus Formulation}
Here, we propose an augmented formulation to study the properties of the proposed optimization protocol \eqref{eq_y}-\eqref{eq_x}. This is similar to the ideas given in consensus literature \cite{Themis_allerton,Themis_delay}. 
Following the notations in \cite{delay_est}, 
define the augmented vectors of size $n(\overline{\tau}+1)$ as $\underline{\mb{d}}^{k} :=[\mb{d}^{k}; \mb{d}^{k-1}; \dots ; \mb{d}^{k-\overline{\tau}}]$, $\underline{\mb{y}}^{k} := \mb{1}_{\overline{\tau}+1} \otimes \mb{y}^{k}$, and $\underline{\mb{x}}^{k} :=[\mb{x}^{k}; \mb{x}^{k-1}; \dots ; \mb{x}^{k-\overline{\tau}}]$. The compact vector form of \eqref{eq_y} is 

\small \begin{align} \label{eq_y_vect}
    \mb{y}^{k+1} & = \operatorname*{argmin}_{\mb{y} \in \Y} \Big(\Phi(\mb{y}) + \mb{y}^\top \pmb{\eta}^{k} + \frac{c}{2} \|\mb{y} - \mb{y}^k + \pmb{\delta}^{k} \|^2 \Big), 
\end{align} \normalsize
and $\underline{\mb{y}}^{k}$ denote the column concatenation of $\mb{y}^{k}$. Vectors $\pmb{\eta}^{k} = [\eta_1^{k};\dots;\eta_n^k]$ and $\pmb{\delta}^{k} = [\delta_1^{k};\dots;\delta_n^k]$ are defined later. Eq.~\eqref{eq_y} is only for notation compactness and, in application, every node $i$ solves minimization problem \eqref{eq_y} locally in its own neighborhood.
Next, define the $n(\overline{\tau}+1) \times n(\overline{\tau}+1)$ matrix $\overline{PW}$ as the augmented version of $W$ via the $0$-$1$ \textit{delay matrix} $P_r$, $r=0,\dots,\overline{\tau}$.

\small \begin{align} \label{eq_aug_WA}
	\overline{PW} = \left( 
	\begin{array}{cccccc}
		P_0 \circ W  & P_1 \circ W  &  \hdots & P_{\overline{\tau}-1} \circ W  & P_{\overline{\tau}} \circ W \\
		\mb{I}_n &   \mb{0}_{n\times n}  &\hdots  & \mb{0}_{n\times n} & \mb{0}_{n\times n}\\
		\mb{0}_{n\times n} & \mb{I}_n &   \hdots  & \mb{0}_{n\times n} & \mb{0}_{n\times n} \\
		\vdots & \vdots &  \ddots & \vdots & \vdots \\
		\mb{0}_{n\times n} & \mb{0}_{n\times n} &  \hdots & \mb{I}_n & \mb{0}_{n\times n}
	\end{array}	
	\right),
\end{align} \normalsize 
with $\circ$ as the element-wise (or Hadamard) product. The non-zero entry in row $i$ and column $j$ of matrix $P_r$ implies that $\mc{I}_{k-r,ij}(r) = 1$, i.e., the associated delay to the link between node $j$ and $i$ is equal to $r$. Assuming the same delay at both sides also implies that $\mc{I}_{k-r,ji}(r) = 1$.

\begin{remark} \label{rem_augP}
The following properties hold for the augmented formulation and the augmented matrix \eqref{eq_aug_WA},
\begin{enumerate}
    \item $P_i$s are $0$-$1$ \textit{symmetric} matrices  satisfying $W = \sum_{i=0}^{\overline{\tau}} (P_i \circ W)$.
\item From the bi-stochasticity of $W$, matrix $\overline{PW}$ is row-stochastic with symmetric bi-stochastic blocks $P_i \circ W$.
\item The entries of $\overline{PW}$ are non-negative.
\end{enumerate}
\end{remark}

Define $\underline{\pmb{\delta}}^{k} := [\pmb{\delta}^{k};\dots;\pmb{\delta}^{k-\overline{\tau}}] = \overline{PW} \underline{\mb{d}}^{k}$ and $\underline{\pmb{\eta}}^{k} := [\pmb{\eta}^{k};\dots;\pmb{\eta}^{k-\overline{\tau}}] = \overline{PW} \underline{\mb{x}}^{k}$. 
Then, the augmented (simplified) version of \eqref{eq_d}-\eqref{eq_x} 
is in the form
\begin{align} \label{eq_d_aug}
	\underline{\mb{d}}^{k+1} &= \underline{\pmb{\delta}}^{k} + \mb{u}^{\overline{\tau}+1}_i \otimes \Xi^m_{i,\overline{\tau}} (\underline{\mb{y}}^{k+1} - \underline{\mb{y}}^{k}), \\ \label{eq_x_aug}
	\underline{\mb{x}}^{k+1} &=  \underline{\pmb{\eta}}^{k} + c  \underline{\mb{d}}^{k+1},
\end{align} 
with the  auxiliary $m \times (\overline{\tau}+1)m$ matrix  $\Xi^m_{i,\overline{\tau}}$  defined as
$\Xi^m_{i,\overline{\tau}} = (\mb{u}^{\overline{\tau}+1}_i \otimes I_m)^\top $
and $\mb{u}^{\overline{\tau}+1}_i$ as the unit column-vector of the $i$'th coordinate ($1\leq i \leq {\overline{\tau}+1}$). 
Then, the augmented form of the initialization is as follows
\begin{align} \label{eq_ini_d}
    \underline{\mb{d}}^{0} = \mb{1}_{\overline{\tau}+1} \otimes \left(\mb{y}^0-\mb{b}\right) = \underline{\mb{y}}^0- \underline{\mb{b}},
\end{align}
with $\mb{b}=[b_1;\dots,b_n]$ and $\underline{\mb{b}} := \mb{1}_{\overline{\tau}+1} \otimes \mb{b}$.
Recall that the variable $\mb{d}$ (and its augmented form $\underline{\mb{d}}$) tracks the feasibility constraint and \eqref{eq_ini_d} assigns the initial feasibility-deviation of the states to $\mb{d}^0$. In this sense, the proposed ADMM strategy is initialization-free. In other words, the algorithm does not need initial values of $\mb{y}^0$ to satisfy the feasibility constraint necessarily, and the solution reaches feasibility over time.  


\subsection{Convergence Analysis}

For the sake of convergence analysis, and following the dynamics~\eqref{eq_d} and \eqref{eq_x}, we first define the average quantities $\overline{d}^k $ and $\overline{x}^k$. Recall that  $\overline{d}^k $ and $\overline{x}^k$ track the mean feasibility-deviation and Lagrangian multiplier over time $k$. Following the \textsf{Parallel-ADMM} formulation in Section~\ref{sec_prob}, $x^k$ is a global variable in ADMM setup and all nodes need to reach consensus on this common variable. Then, following Remark~\ref{rem_augP}

\small \begin{align} \nonumber
    \overline{d}^{k+1} &= \frac{\mb{1}_{n(\overline{\tau}+1)}}{n(\overline{\tau}+1)} (\underline{\mb{y}}^{k} - \underline{\mb{b}}) =\frac{\mb{1}_n^\top}{n} \Xi^m_{i,\overline{\tau}} \overline{PW} (\underline{\mb{y}}^{k+1}- \underline{\mb{b}}) \\ \nonumber
    &= \frac{\mb{1}_n^\top}{n} \Xi^m_{i,\overline{\tau}} \overline{PW} (\underline{\mb{y}}^{k+1}- \underline{\mb{b}} \pm \underline{\mb{y}}^{k}) \\ \nonumber
    &= \overline{d}^{k} + \frac{\mb{1}_n^\top}{n} \Xi^m_{i,\overline{\tau}} \overline{PW} (\underline{\mb{y}}^{k+1} - \underline{\mb{y}}^{k}) \\
    &= \frac{\mb{1}_n^\top}{n} \Xi^m_{i,\overline{\tau}} \overline{PW} (\underline{\mb{y}}^{k+1} - \underline{\mb{b}}) =  \frac{\mb{1}_n^\top}{n} (\mb{y}^{k+1} - \mb{b}),
\end{align} \normalsize
where the last equation follows from the given initialization \eqref{eq_ini_d} and proves by induction as $\overline{d}^{0} = \frac{\mb{1}_n^\top}{n} (\mb{y}^{0} - \mb{b})$ 
and $\overline{\mb{d}}_{\overline{\tau}}^{k+1} $ and $\overline{\mb{x}}_{\overline{\tau}}^{k+1}$ as the augmented vectors of size $n(\overline{\tau}+1)$ with all entries equal to $\overline{d}^{k+1}$ and $\overline{x}^{k+1}$, respectively.
Define the (augmented) optimization error vectors as
\begin{align}
    \underline{\mb{e}}_x^k = \underline{\mb{x}}^k - \overline{\mb{x}}_{\overline{\tau}}^k,~
    \underline{\mb{e}}_d^k = \underline{\mb{d}}^k - \overline{\mb{d}}_{\overline{\tau}}^k,~
    \underline{\mb{z}}^k = \underline{\mb{y}}^k - \overline{\mb{d}}_{\overline{\tau}}^k,
\end{align}
where the last two equations track the (augmented) feasibility error. This is because
\begin{align}
     \overline{\mb{d}}_{0}^{k+1} = \Xi^n_{1,\overline{\tau}} \overline{PW} \underline{\mb{d}}^{k} +  (\Xi^n_{1,\overline{\tau}} - \Xi^n_{2,\overline{\tau}}) \underline{\mb{y}}^{k+1},
\end{align}
and $(\Xi^n_{1,\overline{\tau}} - \Xi^n_{2,\overline{\tau}}) \underline{\mb{y}}^{k+1} = (\Xi^n_{1,\overline{\tau}}\underline{\mb{y}}^{k+1}-\mb{b}) - (\Xi^n_{1,\overline{\tau}}\underline{\mb{y}}^{k}-\mb{b}) $. The error dynamics then follows as
\begin{align} 
    \underline{\mb{e}}_d^{k+1} & = \underline{\mb{d}}^{k+1} - \overline{\mb{d}}_{\overline{\tau}}^{k+1} \label{eq_ed3}\\ 
    & = \overline{PW} \underline{\mb{d}}^{k} + \mb{u}^{\overline{\tau}+1}_1 \otimes (\Xi^n_{1,\overline{\tau}} - \Xi^n_{2,\overline{\tau}}) \underline{\mb{y}}^{k+1} -  \overline{\mb{d}}_{\overline{\tau}}^{k+1} \nonumber\\ \
    & \stackrel{(a)}{=} \overline{PW} \underline{\mb{e}}_d^{k} + \mb{u}^{\overline{\tau}+1}_1 \otimes (\Xi^n_{1,\overline{\tau}} - \Xi^n_{2,\overline{\tau}}) \underline{\mb{y}}^{k+1} - \overline{\mb{d}}_{\overline{\tau}}^{k+1} + \overline{\mb{d}}_{\overline{\tau}}^{k},\nonumber
\end{align}
where, adding $\pm\overline{\mb{d}}_{\overline{\tau}}^{k}$ in step (a), we used the row-stochasticity of $\overline{PW}$ from Remark~\ref{rem_augP}.
For the sake of stability proof analysis, define the \textit{modified} augmented matrix 
\begin{align} \label{eq_pw_tilde}
    \widetilde{PW} = \overline{PW} - \overline{P1},
\end{align}
with $\overline{P1}$ as the augmented version of the stochastic matrix $\frac{1}{n} \mb{1}_{n \times n}$ via the same delay matrix $P$. Similarly, $\widetilde{W} = W - \frac{1}{n} \mb{1}_{n \times n}$.
For the PSD $\widetilde{PW}$, from Frobenius Theorem, $\rho(\widetilde{PW}) \leq \| \widetilde{PW} \| < 1$ since $\|\overline{PW} - \overline{P1}\| < 1$. Define  
$$ \underline{\mb{z}}^{k} =  \mb{u}^{\overline{\tau}+1}_1 \otimes \Xi^n_{1,\overline{\tau}} \underline{\mb{y}}^{k} - \overline{\mb{d}}_{\overline{\tau}}^{k}. $$ 
Then, Eq. \eqref{eq_ed3} can be rewritten as
\begin{align}
    \underline{\mb{e}}_d^{k+1} & =  \widetilde{PW} \underline{\mb{e}}_d^{k} + \underline{\mb{z}}^{k+1} - \underline{\mb{z}}^{k},
\end{align}
since $\overline{P1} \underline{\mb{e}}_d^{k+1} = \overline{P1} (\underline{\mb{d}}^{k+1} - \overline{\mb{d}}_{\overline{\tau}}^{k+1}) = \mb{0}_{n(\overline{\tau}+1)}$, and this added term, only introduced for the sake of proof, has no effect in our algorithm.
Similarly define the average value $\overline{x}^k = \frac{\mb{1}_{n(\overline{\tau}+1)}}{n(\overline{\tau}+1)} \underline{\mb{x}}^k$. Then, from Eq. \eqref{eq_x_aug} and row-stochasticity of $\overline{PW}$, it is easy to show  that $ \overline{x}^{k+1} = \overline{x}^{k} + c\overline{d}^{k}$, and  
\begin{align*}
    \underline{\mb{e}}_x^{k+1} =   \widetilde{PW} \underline{\mb{e}}_x^{k} + c  \underline{\mb{e}}_d^{k+1} = \widetilde{PW} \underline{\mb{e}}_d^{k} + c  \underline{\mb{e}}_d^{k} + c(\underline{\mb{z}}^{k+1} - \underline{\mb{z}}^{k}).
\end{align*}
Rearranging the above terms in its compact form we get 

\small \begin{align} 
	\left( 
	\begin{array}{c}
		\underline{\mb{e}}_x^{k+1} \\
		\underline{c\mb{e}}_d^{k+1}
	\end{array}	\right) &=  
	\left( \begin{array}{cc}
	   \widetilde{PW} & \widetilde{PW} \\
	   \mb{0}	& \widetilde{PW}
	\end{array}	\right) 
	\left( 
	\begin{array}{c}
		\underline{\mb{e}}_x^{k} \\
		c\underline{\mb{e}}_d^{k} 
	\end{array}	\right) + c\left( \begin{array}{c}
	\underline{\mb{z}}^{k+1} - \underline{\mb{z}}^{k} \\
	\underline{\mb{z}}^{k+1} - \underline{\mb{z}}^{k}
	\end{array}	 \right). \label{eq_error_compact}
\end{align} \normalsize
Note that the stability of the error dynamics \eqref{eq_error_compact} is determined by the eigen spectrum of the first row-block of the augmented matrix $\widetilde{PW}$. 
First recall that for bounded $\underline{\mb{x}}^k,\underline{\mb{z}}^k$, from \cite[Lemma~3]{falsone2020tracking}, $\rho(\widetilde{PW}) < 1$ (strict inequality) implies input-to-state stability with respect to $\underline{\mb{z}}^{k+1} - \underline{\mb{z}}^{k}$, i.e., bounded and stable errors $\underline{\mb{e}}^k_d$ and $\underline{\mb{e}}^k_x$ if $\underline{\mb{z}}^{k+1} - \underline{\mb{z}}^{k}$ is bounded and converges to $0$ over time $k$. Note that since $|y^k_i|$ is bounded by $m_i,M_i$ (the box-constraint), hence follows the boundedness of $\underline{\mb{z}}^{k}$ and the errors $\underline{\mb{e}}^k_d$ and $\underline{\mb{e}}^k_x$.

\begin{theorem}[Convergence and Optimality] \label{thm_conv}
Let $\mb{y}^*$ and $x^*$ denote the optimal solution of primal and dual problems \eqref{problem_initial_0} and \eqref{problem_initial}, respectively. Under Assumptions \ref{assup_convex}-\ref{ass_delay} and Algorithm~\ref{alg_ac},
\begin{enumerate}
    \item The error norms $\|\underline{\mb{e}}_d^{k}\|$ and  $\|\underline{\mb{e}}_x^{k}\|$, along with the feasibility-deviaiton term $\overline{d}^k$ converge to zero.
    \item The sequence $\{\|\underline{\mb{x}}^{k} - \underline{\mb{x}}^* \|^2 + c^2 \|\underline{\mb{z}}^{k} - \underline{\mb{y}}^* \|^2 \}$ is bounded and convergent (with $\underline{\mb{x}}^* := \mb{1}_{n(\tau + 1)} x^*$).
    \item The primal-dual optimization variables satisfy $\lim_{k\rightarrow \infty} \mb{y}^k = \mb{y}^* $ and $\lim_{k\rightarrow \infty} \underline{\mb{x}}^k = \underline{\mb{x}}^*$.
\end{enumerate}
\end{theorem}
\begin{proof}
 See the Appendix.
\end{proof}

\begin{remark} \label{rem_goodnews}
    The good news about \textsf{DTAC-ADMM} is that it is initialization-free as, $\underline{\mb{d}}^{0} \neq \mb{0}$ in Eq.~\eqref{eq_ini_d} implies that the algorithm converges for $\mb{y}(0) \neq \mb{b}$ in general; although, addressing uniform-connectivity, dynamic network topology, and possible model nonlinearities are still open for further research in this scenario.
\end{remark}

\begin{remark} \label{rem_ax}
The problem can be generalized to \textit{weighted} coupling-constraints as
\begin{equation}\label{problem_initial_ax}
	\begin{aligned}
		\operatorname*{min}_{\mb{y}  \in \mathbb{R}^{n} }  & \sum_{i=1}^{n} \Phi(\mb{y})  =  \operatorname*{min}_{y_i\in \mathbb{R} }  \sum_{i=1}^{n} \phi_i(y_i), \\
		\mbox{subject~to} &\sum_{i=1}^n (a_i y_i - b_i)=0,	~ y_i \in \Y_i
	\end{aligned}
\end{equation}
with $a_i \in \mathbb{R}$ as a weighting factor on the state $y_i$. Then, the solution can be adjusted accordingly by substituting $a_i y_i$, $a_i y_i^k$, and $a_i y_i^{k+1}$ in the right-hand-side of the Eq.~\eqref{eq_y} and~\eqref{eq_d} (the cost $\phi_i(y_i)$ in \eqref{eq_y} remains the same).  
An example application is given in the next section.
\end{remark}



\section{Application: Scheduling over Energy Grid} \label{sec_app}
Consider the EDP for a group of $n$ generators tasked to produce power according to a fixed demand over a specific time interval. The demand $b$ may change over certain time intervals, and, thus, one key point is the so-called solution feasibility, referring to the global constraint $\sum_{i=1}^n y_i = b$ to be held at the termination point, implying that the demand $b$ need to meet the produced power $\sum_{i=1}^n y_i$.
The solution can be either (i) all-time feasible or (ii) gain feasibility sufficiently fast before the termination time. Then, the optimal solution is assigned to the generators, optimizing the following power generation cost model
\begin{align} \label{eq_f_quad_edp}
\min_y  &\sum_{i=1}^n \gamma_i y_i^2+ \beta_i y_i + \alpha_i,\\ ~\mbox{subject~to}~&\sum_{i=1}^n y_i = b, ~~ m_i \leq y_i \leq M_i
\end{align}
where parameters $\gamma_i,\beta_i,\alpha_i$ follow the type of the generator $i$ (i.e., oil-fired, coal-fired, etc).

As mentioned in remark~\ref{rem_goodnews}, \textsf{DTAC-ADMM} algorithm is initialization-free and for any non-feasible initial power generation in $[m_i~M_i]$ it can reach the demand value $b = \sum_{i=1}^n b_i$. This is better illustrated in Fig.~\ref{fig_delay} and is an improvement over the existing coordination methods, e.g.,  \cite{cherukuri2015distributed}. We consider a network of $n = 6$ generators connected over a simple cycle network with self weights equal to $0.5$ and other link-weights equal to $0.25$. We run the algorithm over $10^4$ iterations (stopping criteria). Note that in some literature, e.g. \cite{banjac2019decentralized}, the local box-constraints $m_i \leq y_i \leq M_i$ are not addressed, while our solution perfectly handles these local constraints as shown in Fig.~\ref{fig_delay} (primal states). The case of $\overline{\tau} = 0$ refers to the no-delay solution \textsf{Distributed-Parallel-ADMM} by \cite{falsone2020tracking}. Evidently, the performance of our \textit{delay-tolerant} \textsf{DTAC-ADMM} is comparable with \cite{falsone2020tracking} while, further, handling network latency. 
\begin{figure}
	\centering
 	\includegraphics[width=1.67in]{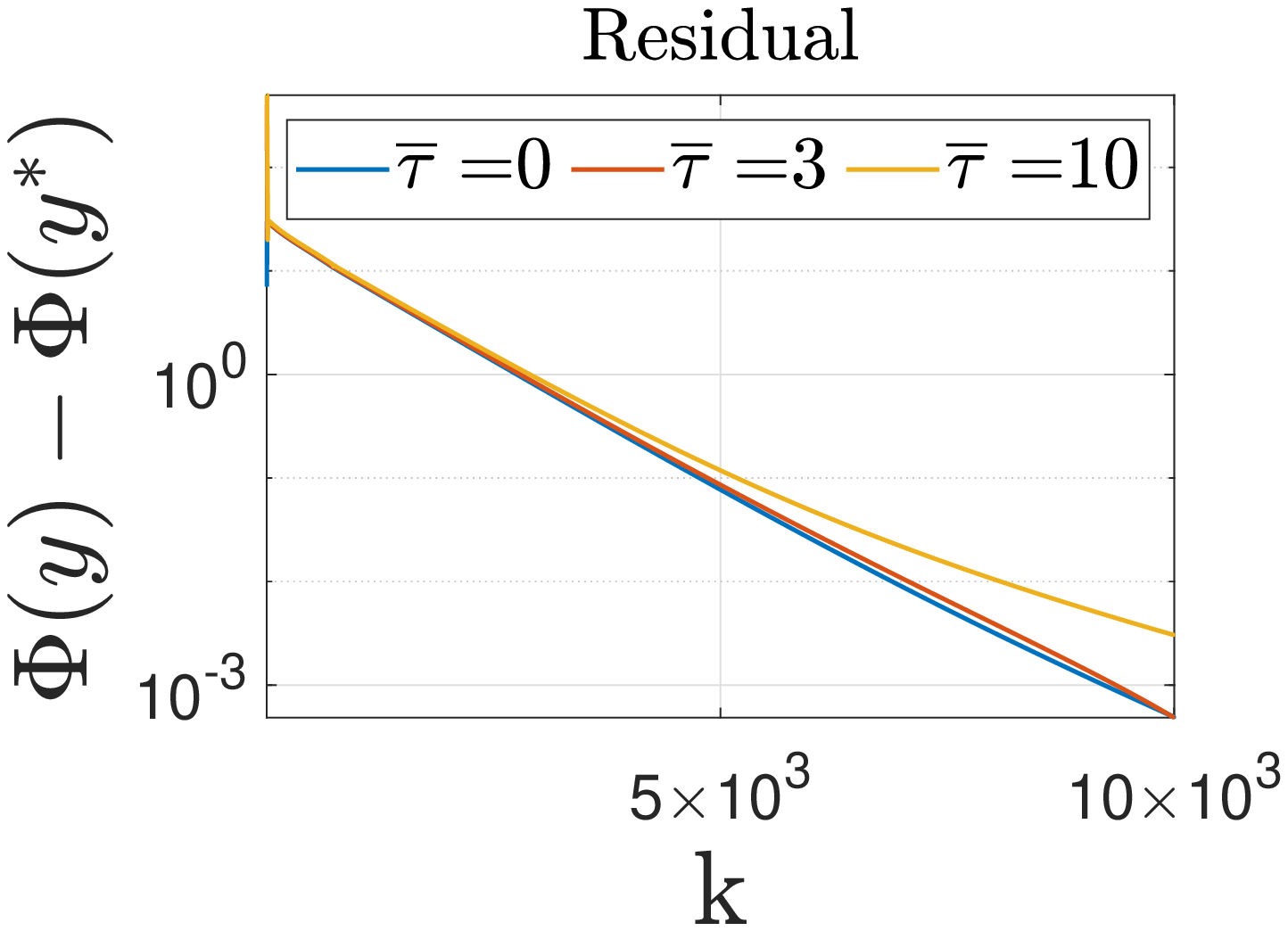}
 	\includegraphics[width=1.67in]{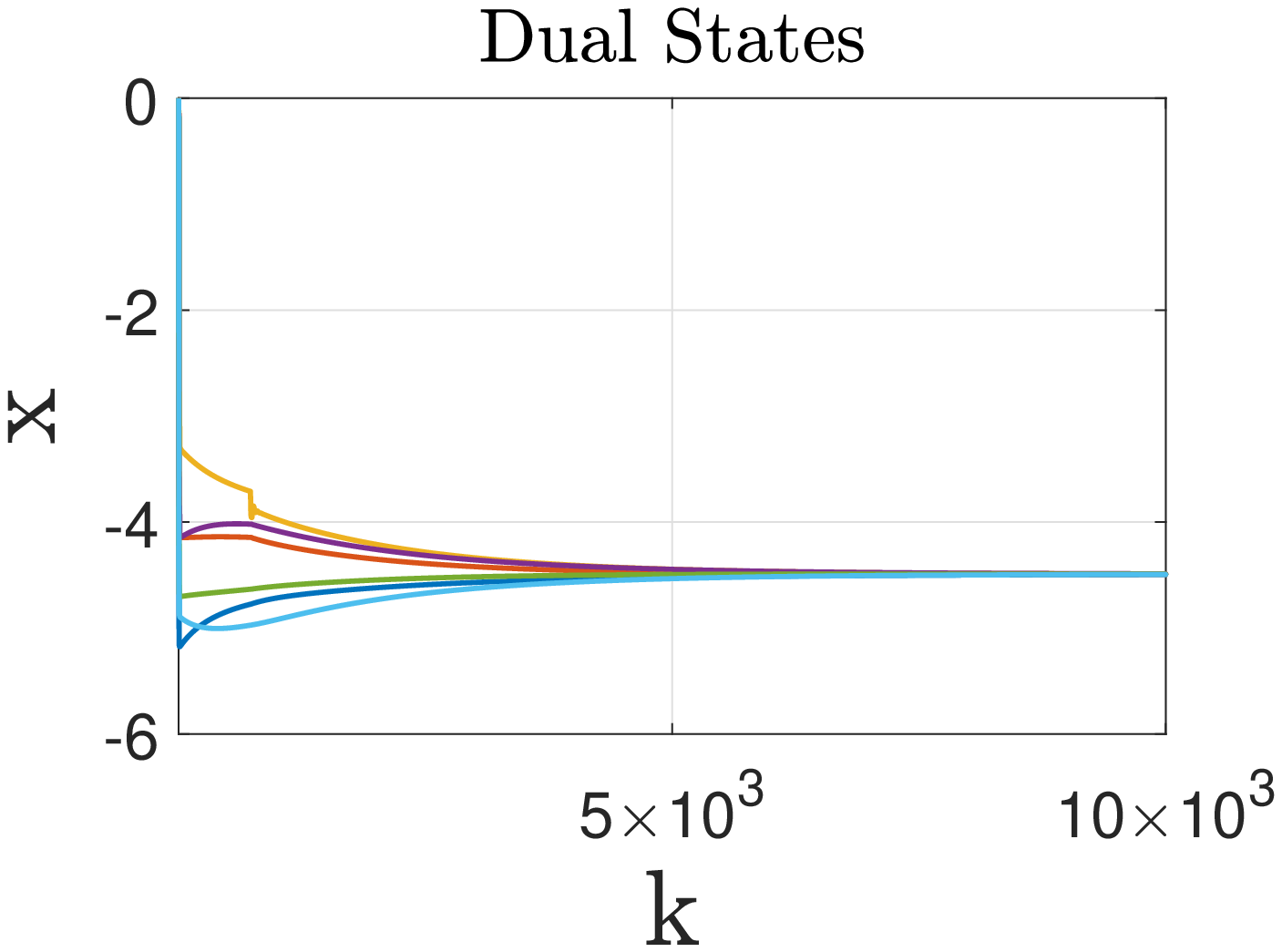}
 	\includegraphics[width=1.67in]{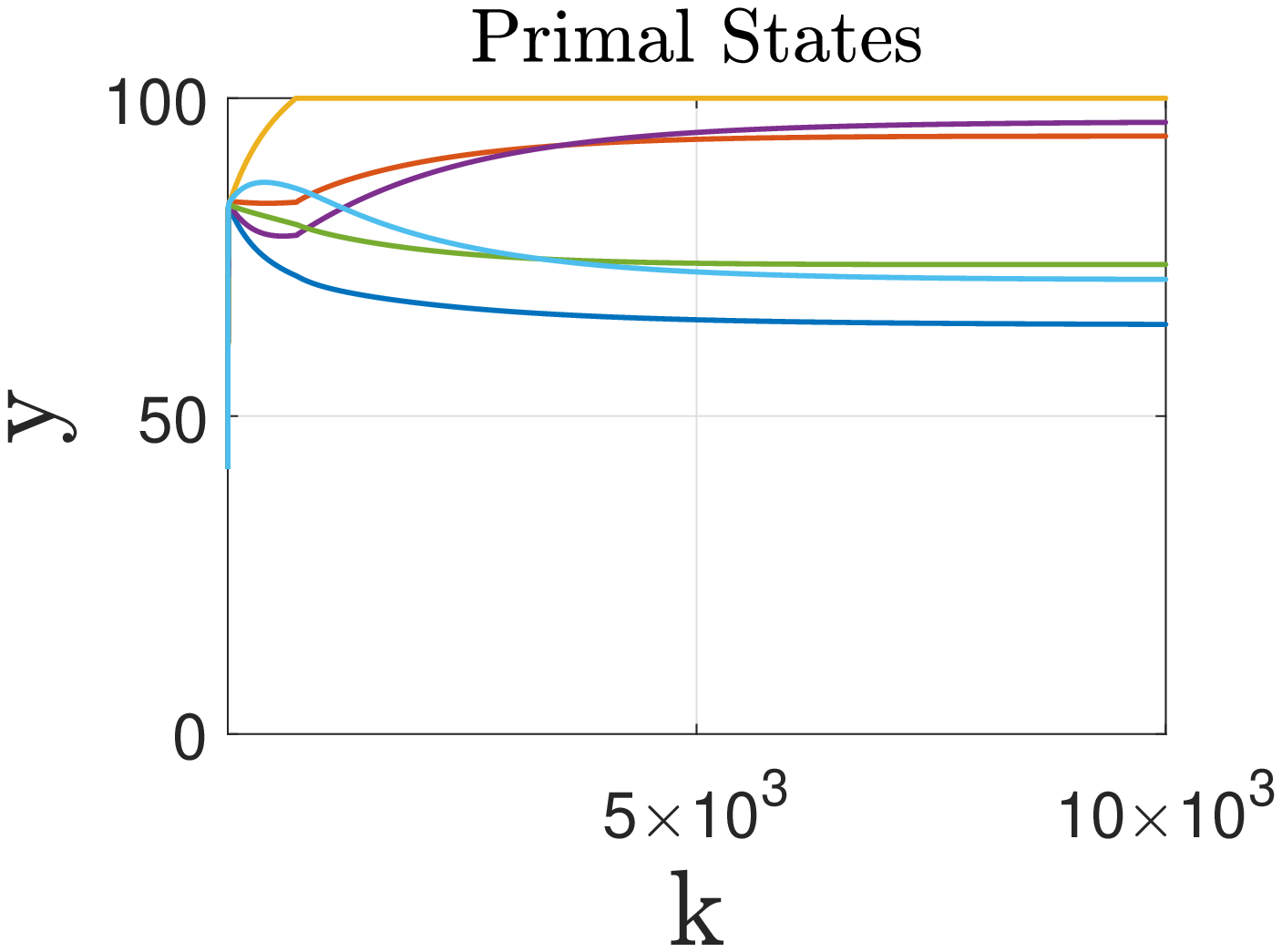}
 	\includegraphics[width=1.67in]{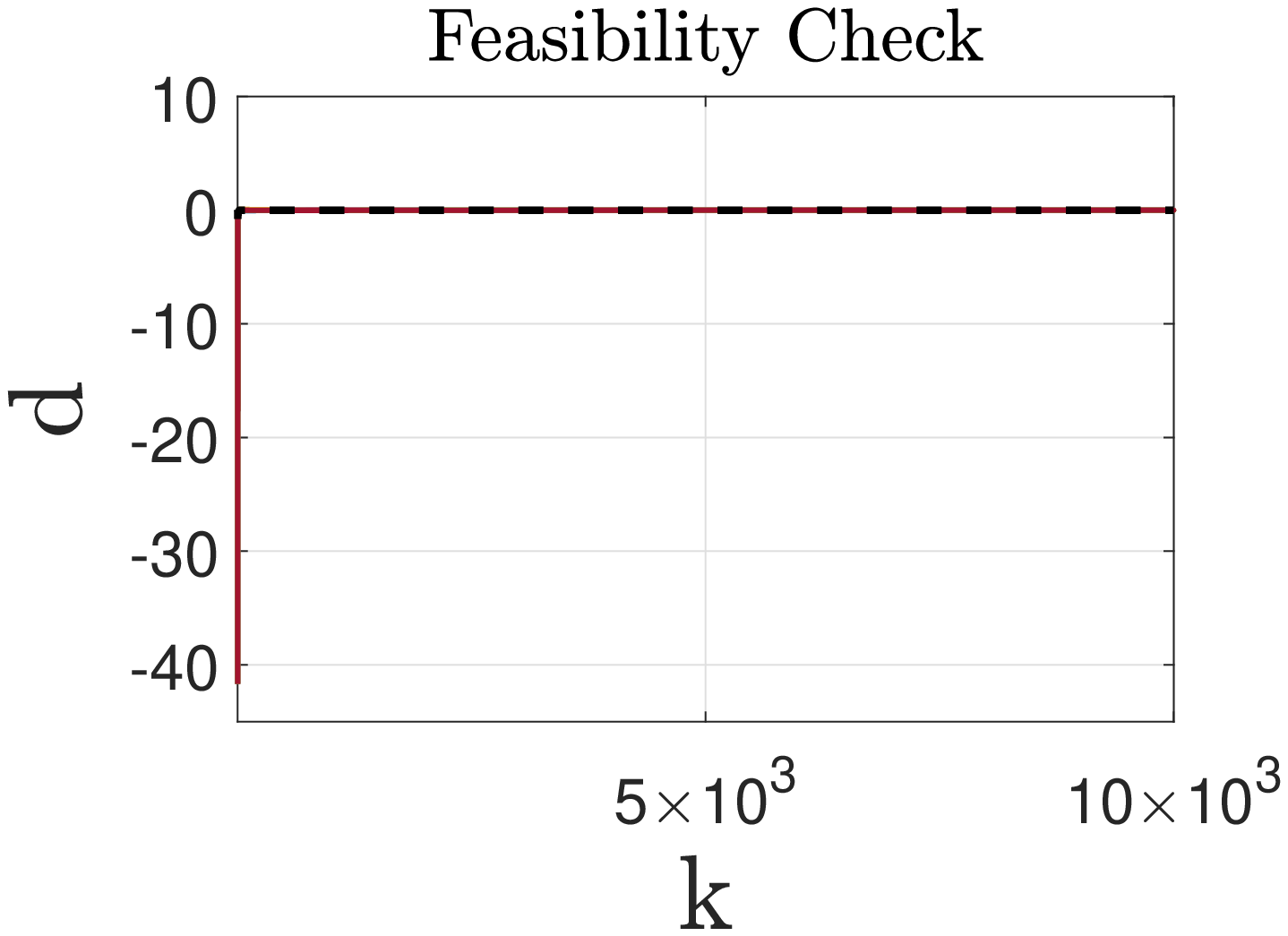}
    \caption{The performance of \textsf{DTAC-ADMM} algorithm is shown in this figure. The parameters are: $b=500$, $m_i=0$, $M_i = 100$, $c=5$, $y_i^0 = \frac{b}{2n}=41.7$. The cost parameters are chosen randomly. The bound on the delays are $\overline{\tau} = 3,10$. The states $y_i$ remain inside the box-constraints, dual states $x_i$ reach consensus on $x^*$, and feasibility-deviations $d_i$ converges to zero. The dashed black line represents $\overline{d}^k$ which converges sufficiently fast to reach feasibility. }
    \label{fig_delay}
\end{figure}

Following the formulation given by \cite{vrakopoulou2017chance,mrd_vtc2022}, next we consider a more general problem: a combination of (thermal) batteries for energy reservation and the generators for energy production over the grid. The cost model then changes to  
\begin{align} \label{eq_f_quad_edp2}
\min_{y,w}  &\sum_{i=1}^{n_1} \gamma_i y_i^2 + \beta_i y_i +  \alpha_i + \sum_{i=1}^{n_2} \zeta_i w_i, \\ ~\mbox{subject~to}~&\sum_{i=1}^{n_1} y_i = b + \sum_{i=1}^{n_2} w_i, ~~ m_i \leq y_i,w_i \leq M_i
\end{align}
with $w_i$ denoting the reserved power (state) of battery $i$. The feasibility constraint changes accordingly: the sum of produced power $\sum_{i=1}^{n_1} y_i$ by $n_1=4$ generators is equal to the sum of reserved power $\sum_{i=1}^{n_2} w_i$ by $n_2=2$ batteries plus the demand $b$. Then, the problem can be formulated as in Remark~\ref{rem_ax}, with $a_i = 1$ for the generators and $a_i = -1$ for the batteries and $b_i = \frac{b}{n_1+n_2}$. The simulation is shown in Fig.~\ref{fig_delay2} for a 2-hop network with $0.2$ link weights and random $M_i$ and $m_i$. The elapsed time of the MATLAB R2021b simulation on an Intel Core i5 @ 2.40GHz processor RAM 8GB is $0.65$ sec. The initial value of $|\overline{d}^0|=5.64$ reduces by order of $10^{-6}$ after $2400$ iterations ($0.156$ sec), implying that, in practical applications, the feasibility can be gained sufficiently fast. 
\begin{figure}
	\centering
 	\includegraphics[width=1.67in]{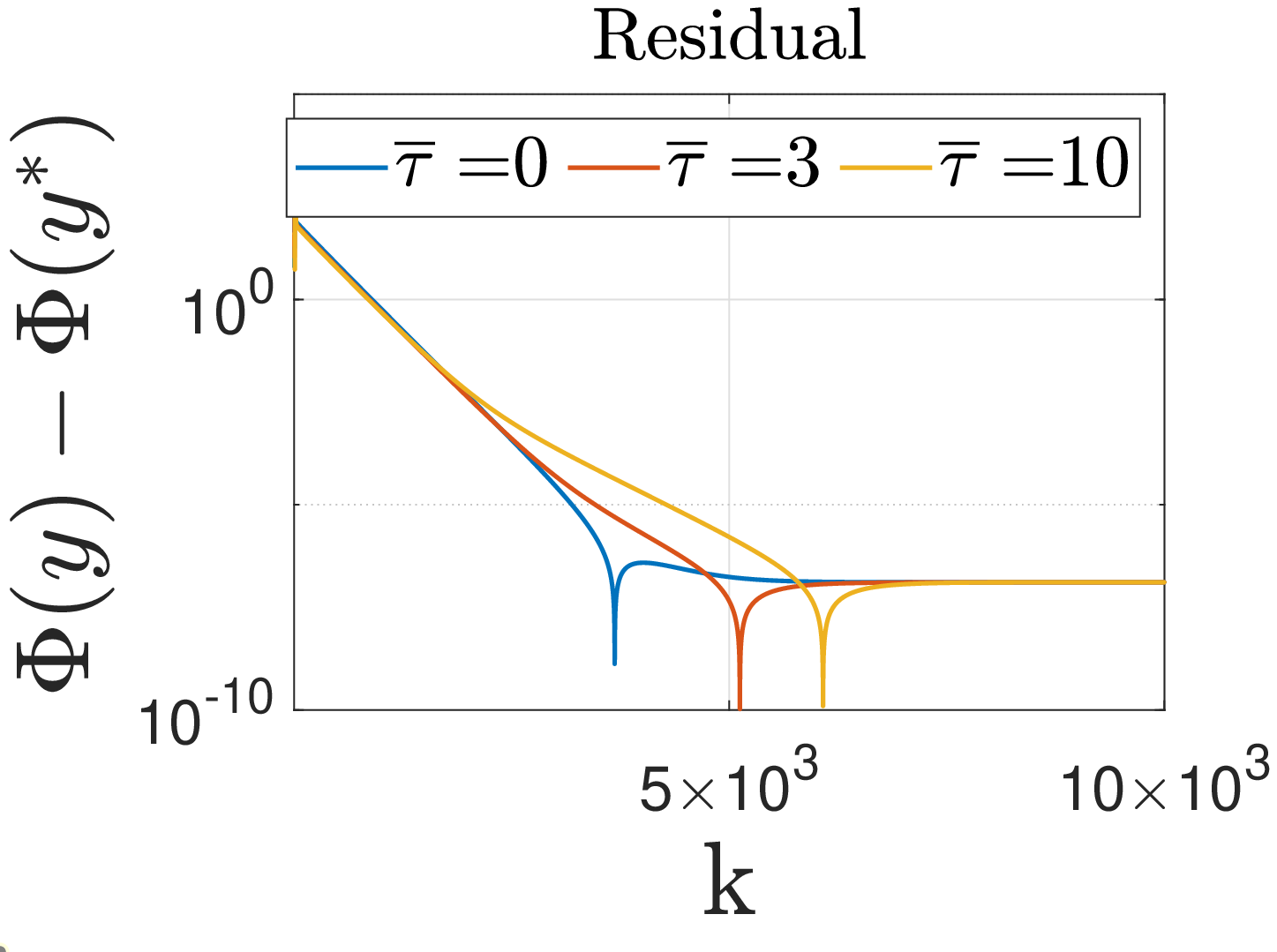}
 	\includegraphics[width=1.67in]{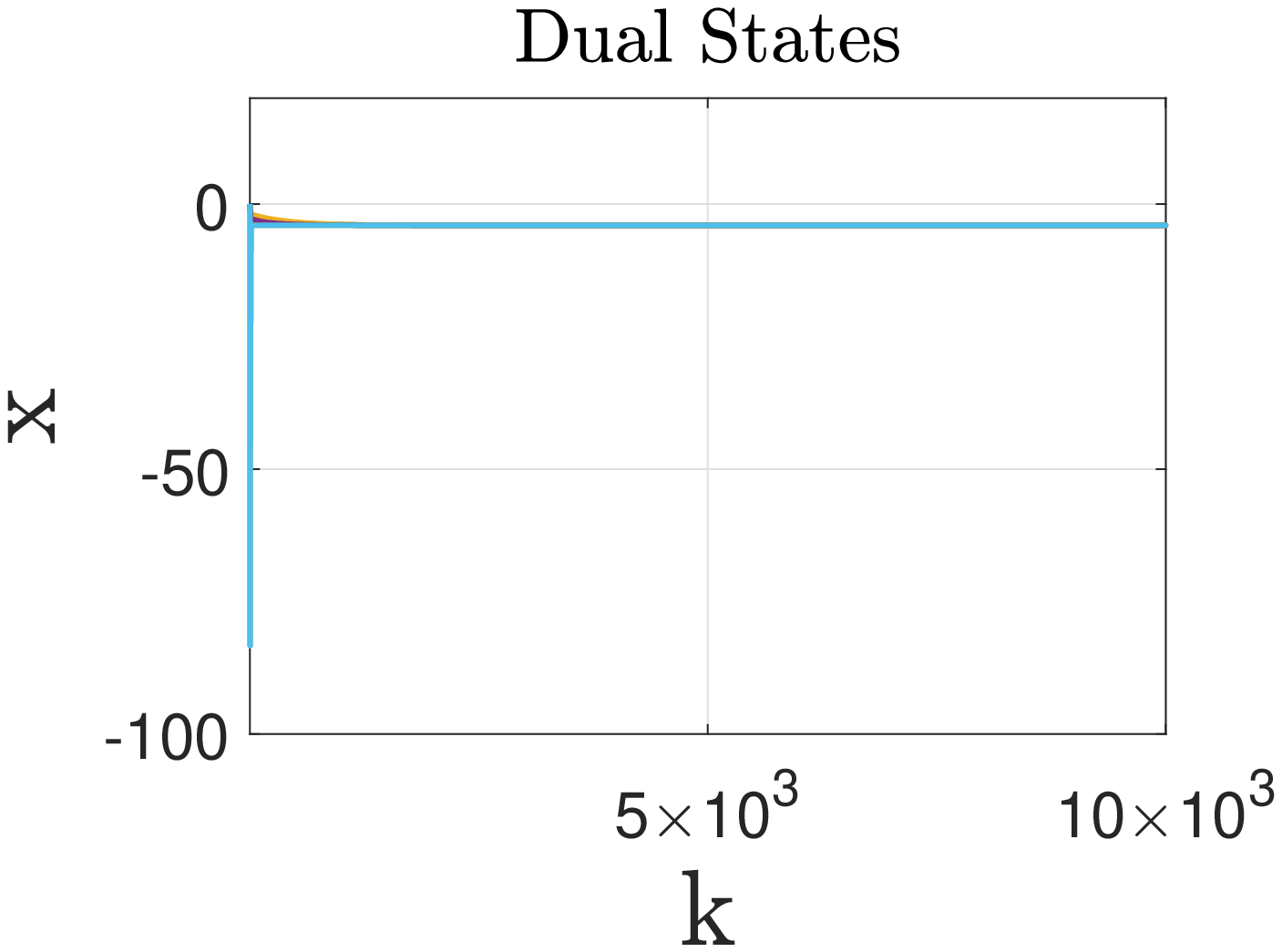}
 	\includegraphics[width=1.67in]{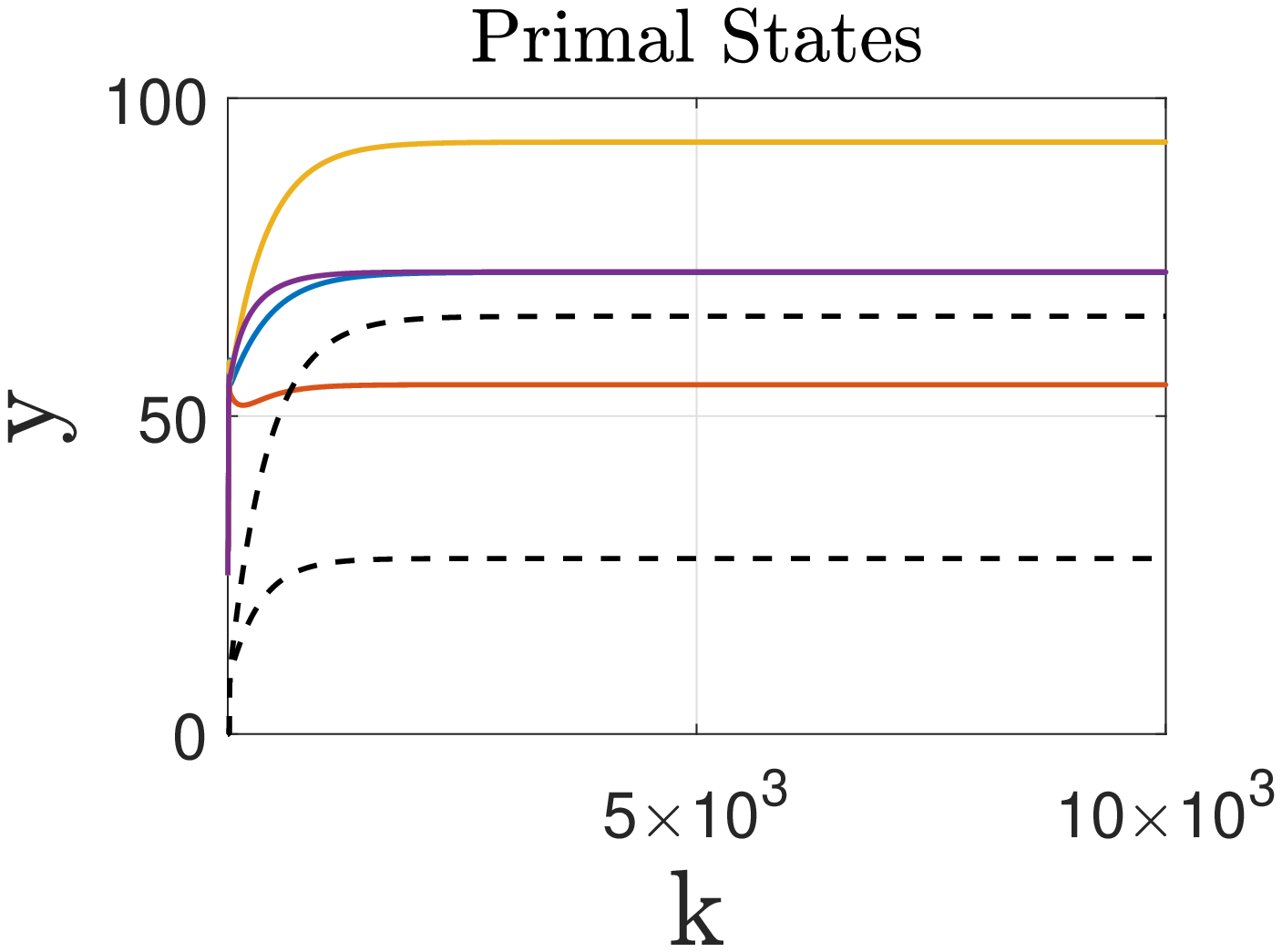}
 	\includegraphics[width=1.67in]{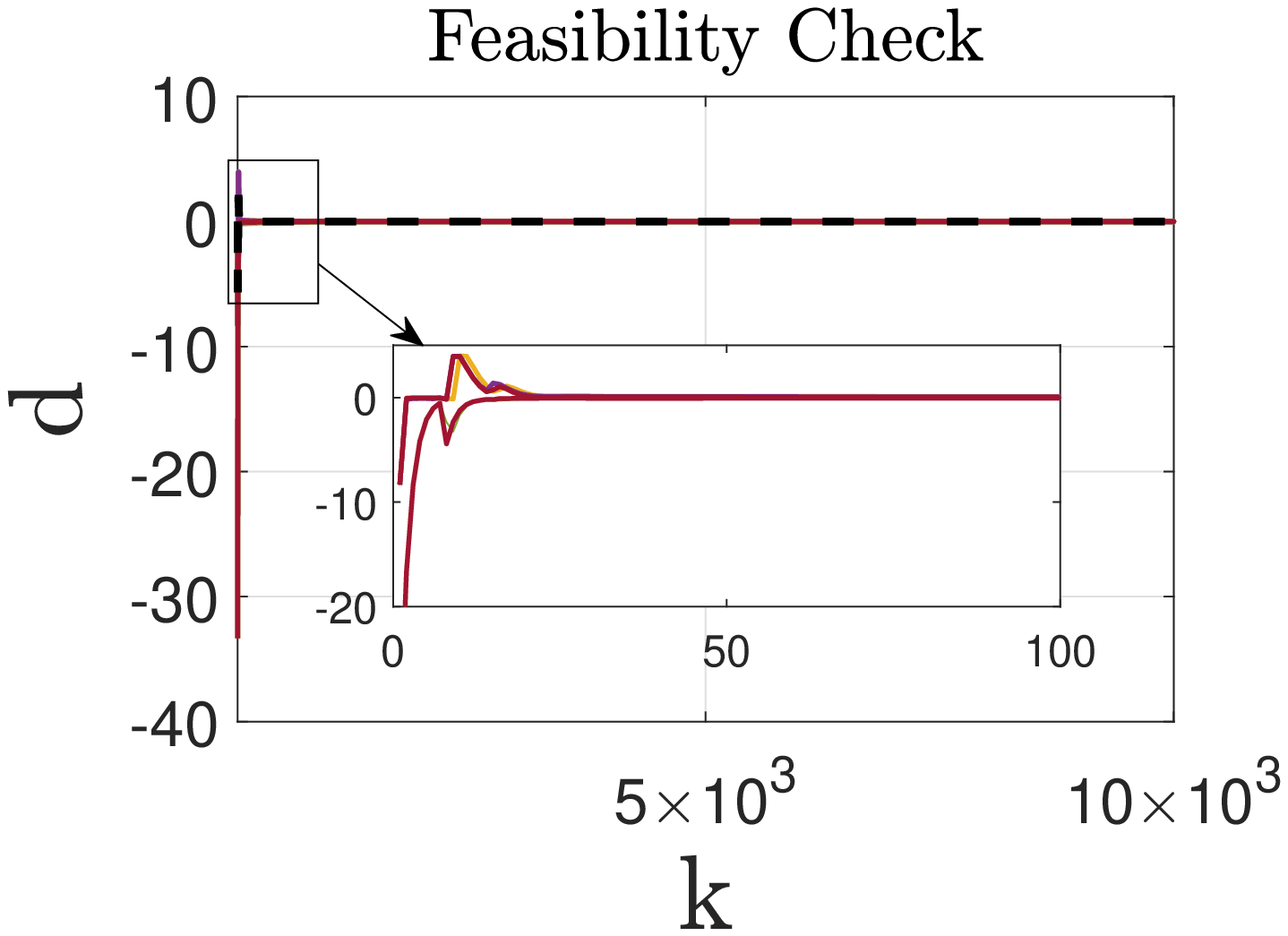}
    \caption{The performance of battery-generator scheduling for similar setup as in Fig.~\ref{fig_delay}. The parameters $\gamma_i$ are randomly chosen in the range $[0.02~0.05]$ and $\beta_i$ in $[-4~1]$ for the generators and $-4$ for the batteries. The dashed and solid lines represent the (primal) $y$ states of the batteries and generators. The feasibility constraint is met for $b=200$. }
    \label{fig_delay2}
\end{figure}

\section{Simulation: Non-Quadratic Cost} \label{sec_sim}
In this section, we consider a logarithmic cost model in a similar setup as in \cite{fast} with parameters chosen similarly,
\begin{eqnarray} \label{eq_F2}
\phi_i(y_i) = \frac{1}{2}a_i(y_i-c_i)^2 + \log(1+\exp(b_i(y_i-d_i))).
\end{eqnarray}
Note that, for this logarithmic function, Eq. \eqref{eq_y_vect} cannot be solved in closed form. We use MATLAB \texttt{fsolve} command in this step. The simulation results are shown in Fig.~\ref{fig_nonquad} compared with \cite{fast}. \textsf{DTAC-ADMM} algorithm shows acceptable performance in the presence of time-invariant delays.
\begin{figure}
	\centering
 	\includegraphics[width=1.67in]{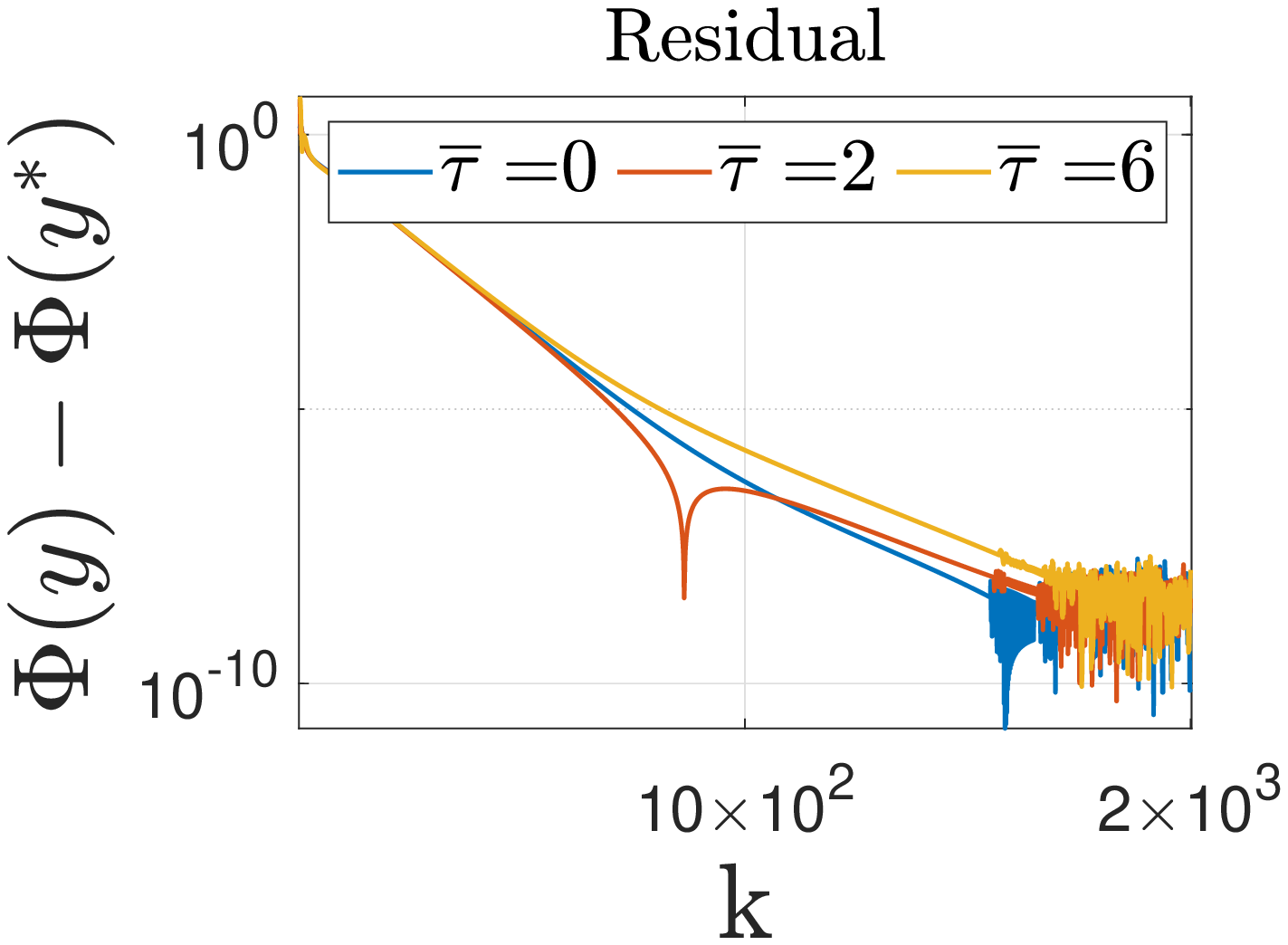}
 	\includegraphics[width=1.67in]{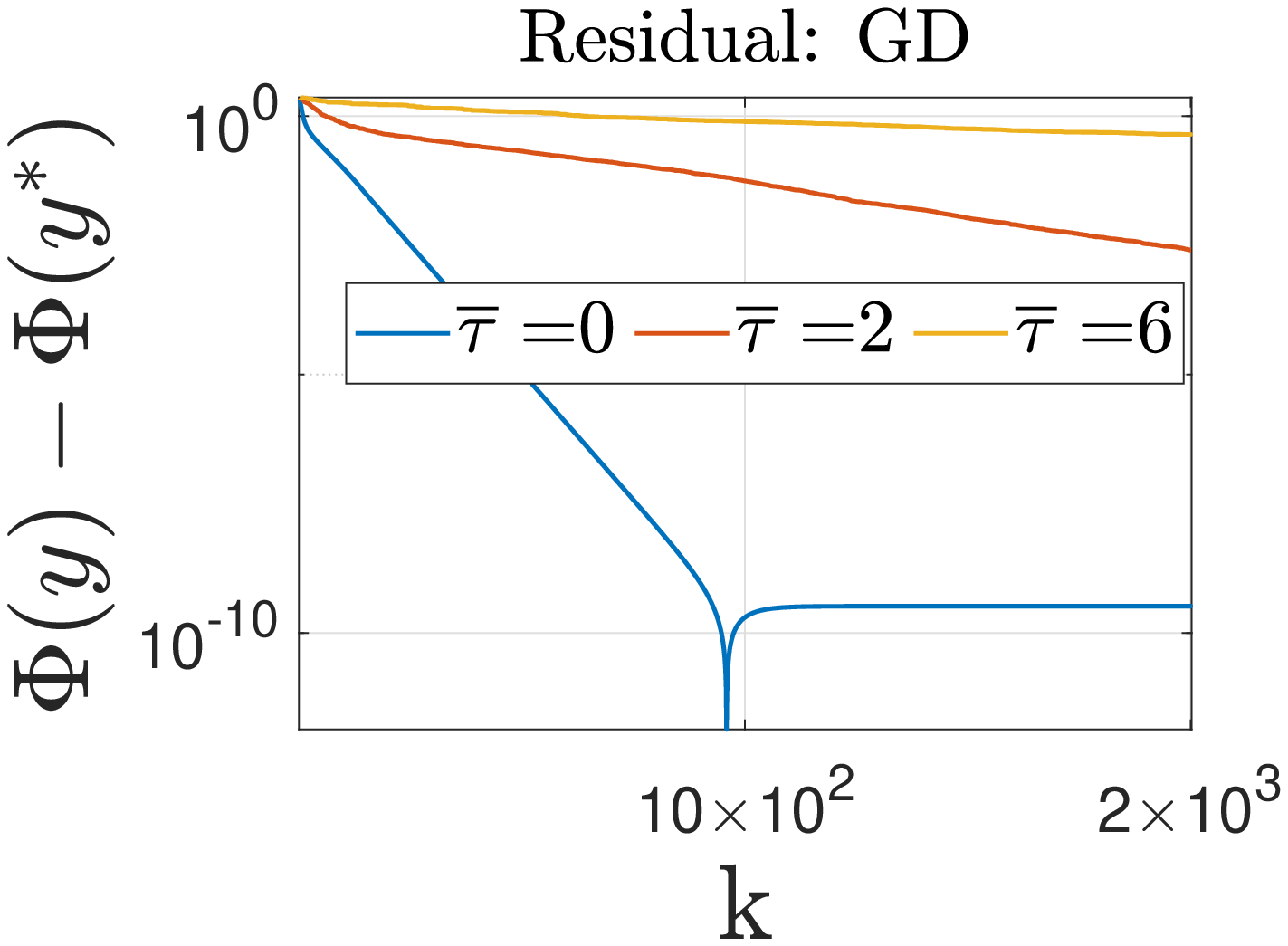}
 	\includegraphics[width=1.67in]{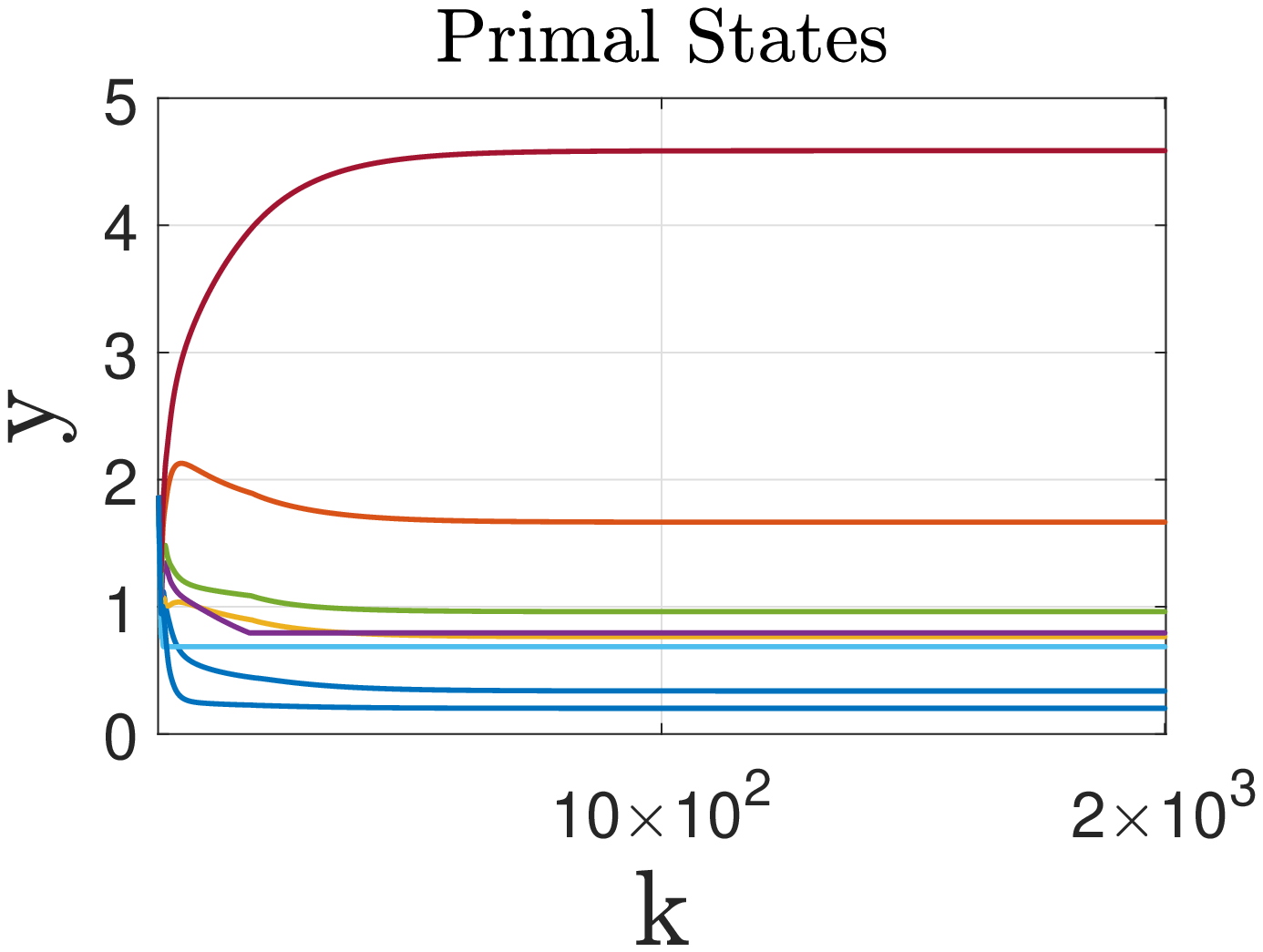}
 	\includegraphics[width=1.67in]{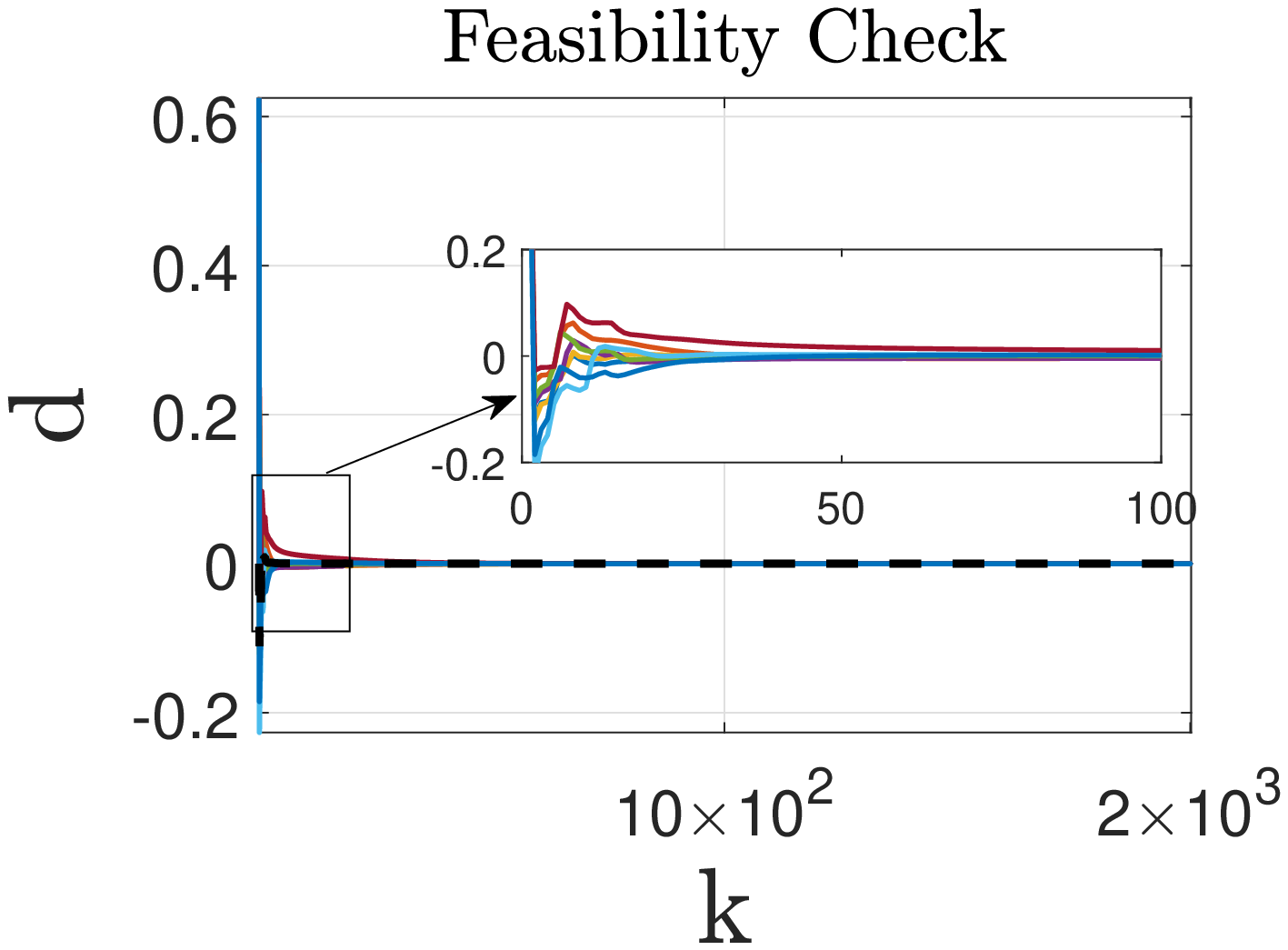}
    \caption{The performance of the \textsf{DTAC-ADMM} is presented for the non-quadratic cost model \eqref{eq_F2} with $y^0_i = \frac{b}{2n}$. The Top-Right figure shows the performance of GD-based solution proposed by \cite{fast} assuming \textit{primal feasible} initialization $y^0_i = \frac{b}{n}$ (as discussed in Section~\ref{sec_intro}). The parameters of the GD solution are: $\eta=0.05$, $v_1 = v_2 = 1$, and no quantization. Additive penalty terms are used to handle box-constraints in \cite{fast}.  }
    \label{fig_nonquad}
\end{figure}
%
%
%
%
\section{Conclusions and Future  Directions}\label{sec_con}

This paper proposes \textsf{DTAC-ADMM} algorithm for distributed resource allocation and advances the existing methods by addressing heterogeneous time delays over the network. 
Relaxing the all-time network connectivity assumption to uniform connectivity over directed networks is a promising future research direction. Further, addressing possible model nonlinearities, e.g., quantization \cite{rikos2021optimal,magnusson2018communication} or clipping \cite{mrd20211st}, in the ADMM-based formulation is another interesting direction.

\section*{Appendix}
In order to prove Theorem~\ref{thm_conv}, we first prove the results for the homogeneously-delayed protocols \eqref{eq_y_sync}-\eqref{eq_x_sync} and then extend it to the heterogeneous case \eqref{eq_y}-\eqref{eq_x}. The matrices $\widetilde{W}$ and $\widetilde{PW}$ follow the definitions in Section~\ref{sec_async}.

\begin{lemma} \label{lem_bertsek}
    The following holds under Assumptions \ref{assup_convex}-\ref{ass_dual} and dynamics \eqref{eq_y_sync}-\eqref{eq_x_sync}.
    \begin{align} \label{eq_bertsek_2}
        \phi_i(y_i^{k+1}) &+ x_i^{k+1} (y_i^{k+1} - y_i^*)  \leq 
        \phi_i(y^*_i). 
    \end{align}
\end{lemma}
\begin{proof}
Following \cite[Lemma~4.1]{bertsekas1989parallel} and Assumptions \ref{assup_convex}-\ref{ass_dual} we have
        \begin{align} \nonumber
        \phi_i(y_i^{k+1}) &+ \Big(\eta_i^{k} y_i +  c (y_i^{k+1} - y_i^k + \delta_i^{k} ) \Big) y_i^{k+1} \leq \\ \label{eq_bertsek}
        & \phi_i(y_i) + \Big(\eta_i^{k} y_i +  c (y_i - y_i^k + \delta_i^{k} ) \Big) y_i,
    \end{align}
for all $y_i \in \Y_i$. Substituting $y_i = y_i^*$ and Eq.~\eqref{eq_d} and \eqref{eq_x} the proof follows.  
\end{proof}
\begin{lemma} \label{lem_local_opt}
    The following holds under Assumptions \ref{assup_convex}-\ref{ass_net} for the protocol \eqref{eq_y_sync}-\eqref{eq_x_sync}
    \begin{align} \nonumber 
        \| \overline{\mb{x}}^{k+1} - \mb{x}^* \|^2 &+ 2c \| \mb{z}^{k+1} - \mb{y}^* \|^\top \mb{e}^{k+1}_x \leq \\ \label{eq_loc_opt}
        & \| \overline{\mb{x}}^{k-\overline{\tau}} - \mb{x}^* \|^2 - \| \overline{\mb{x}}^{k+1} - \overline{\mb{x}}^{k-\overline{\tau}} \|^2,
    \end{align}
\end{lemma}
with $\mb{z}^{k+1} = \mb{y}^{k+1}- \overline{\mb{d}}^{k+1} $.

\begin{proof}
We provide a general sketch of the proof, see details for the no-delay case in \cite[Proposition~1]{falsone2020tracking}. Using the saddle point inequality for primal variable $\mb{y}$ and dual variable $x$ in \eqref{problem_initial}, we have $ L(\mb{y}^*,x) \leq  L(\mb{y}^*,x^*) \leq  L(\mb{y},x^*)$. Recalling that $L(\mb{y}^*,x^*) = \Phi(\mb{y}^*)$. Using \eqref{eq_bertsek_2} in Lemma~\ref{lem_bertsek} over the time-scale $k+1-\overline{\tau}$ to $k+1$,
    \begin{align} \nonumber
         &\Phi(\mb{y}^{k+1}) +   (\mb{y}^{k+1} - \mb{y}^*)^\top \mb{x}^{k+1} \leq \Phi(\mb{y}^*)  \\
         & \leq \Phi(\mb{y}) +  (\mb{y}^{k+1} - \mb{b})^\top \mb{x}^* = \Phi(\mb{y}) +  (\mb{y}^{k+1} - \mb{y}^*)^\top \mb{x}^*,
         \label{eq_loc_opt_proof}
    \end{align} 
where the last equality follows the feasibility of $\mb{y}^*$. Putting $\mb{y} = \mb{y}^{k+1}$ and simplifying by algebraic manipulations, $(\mb{y}^{k+1} - \mb{y}^*)^\top (\overline{\mb{x}}^{k+1} - \mb{x}^*) + (\mb{y}^{k+1} - \mb{y}^*)^\top (\mb{x}^{k+1} - \overline{\mb{x}}^{k+1}) \leq 0$. From the definitions and using $\mb{1}_n^\top \mb{y}^* = \mb{b}$, it follows that $(\mb{y}^{k+1} - \mb{y}^*)^\top (\overline{\mb{x}}^{k+1} - \mb{x}^*) = \frac{1}{c}(\overline{\mb{x}}^{k+1} - \overline{\mb{x}}^{k-\overline{\tau}})^\top (\overline{\mb{x}}^{k+1} - \mb{x}^*)$, $(\mb{y}^{k+1} - \mb{y}^*)^\top (\mb{x}^{k+1} - \overline{\mb{x}}^{k+1}) = (\mb{y}^{k+1} - \mb{y}^*)^\top (\mb{z}^{k+1} - \mb{y}^*)$, and the proof follows from the algebraic equality 
\begin{align} \nonumber
    2 (\overline{\mb{x}}^{k+1} &-\mb{x}^*)^\top(\overline{\mb{x}}^{k+1}-\overline{\mb{x}}^{k-\overline{\tau}}) = \|\overline{\mb{x}}^{k+1}-\mb{x}^*\|^2 \\ 
    &+ \|\overline{\mb{x}}^{k+1}-\overline{\mb{x}}^{k-\overline{\tau}}\|^2 - \|\mb{x}^*-\overline{\mb{x}}^{k-\overline{\tau}}\|^2,
\end{align}\normalsize  

\end{proof}

\begin{corollary} \label{cor_xyz}
    Let Assumptions \ref{assup_convex}-\ref{ass_delay} hold and assume $\underline{\mb{d}}^k = \mb{1}_{\overline{\tau}+1} \otimes \mb{d}^k$, $\underline{\mb{x}}^k = \mb{1}_{n(\overline{\tau}+1)} \otimes \overline{x}^k$, and $\underline{\mb{y}}^k = \mb{1}_{\overline{\tau}+1} \otimes \mb{y}^k$ (e.g., as $k\rightarrow \infty$). Then, Eq.~\eqref{eq_loc_opt} can be reformulated as
       \begin{align} \nonumber 
        \| \overline{\mb{x}}_{\overline{\tau}}^{k+1} - \underline{\mb{x}}^* \|^2 &+ 2c \| \underline{\mb{z}}^{k+1} - \underline{\mb{y}}^* \|^\top \underline{\mb{e}}^{k+1}_x \leq \\ \label{eq_loc_opt2}
        & \| \overline{\mb{x}}_{\overline{\tau}}^{k} - \underline{\mb{x}}^* \|^2 - \| \overline{\mb{x}}_{\overline{\tau}}^{k+1} - \overline{\mb{x}}_{\overline{\tau}}^{k} \|^2,
    \end{align}
    with $\underline{\mb{z}}^{k+1} = \underline{\mb{y}}^{k+1}- \overline{\mb{d}}_{\overline{\tau}}^{k+1}$.
\end{corollary}
As discussed later in the proof of Theorem~\ref{thm_conv}, the above corollary relates to the augmented formulation via dynamics \eqref{eq_y_vect}-\eqref{eq_x_aug} in limit as $k\rightarrow \infty$.

\begin{lemma}[Homogeneous Delay] \label{lem_synch}
    Let Assumption \ref{ass_net} hold and assume arbitrary homogeneous delay $\overline{\tau}$. The modified augmented matrix $\widetilde{PW}$ in \eqref{eq_pw_tilde} satisfies $\rho(\widetilde{PW}) = \rho(\widetilde{W})^{\frac{1}{\overline{\tau}+1}} <1$. 
\end{lemma} 
\begin{proof}
    Let $\sigma_{\widetilde{W}}(\lambda)$ represent the characteristic polynomial of the matrix  $\widetilde{W}$. From \cite[Lemma~4]{delay_est},  
    \begin{align}
        \sigma_{\widetilde{PW}}(\lambda) = \sigma_{\widetilde{W}}(\lambda^{\overline{\tau}+1}).
    \end{align}
    Let $\lambda_i$, $i=1,\dots,n$ denote the roots of $\sigma_{\widetilde{W}}(\lambda) =  
    0$; then, $\lambda_i^{\overline{\tau}}$ denote the roots of $\sigma_{\widetilde{PW}}(\lambda) =  
    0$. This implies that $\rho(\widetilde{PW}) = \rho(\widetilde{W})^{\frac{1}{\overline{\tau}+1}} < 1$. 
\end{proof}

\begin{lemma}  \label{lem_conv_sync}
    Under Assumptions~\ref{assup_convex}, \ref{ass_feasible}, and \ref{ass_net} with homogeneous delays $\overline{\tau}$ under dynamics \eqref{eq_y_sync}-\eqref{eq_x_sync} the following error-dynamics is convergent.
    \begin{align} \label{eq_F}
        \mb{e}^{k+1} = F \mb{e}^{k-\overline{\tau}} +  \mb{1}_2 \otimes c(\mb{z}^{k+1} - \mb{z}^{k-\overline{\tau}}),
    \end{align}
    with
   \begin{align}
      F = \left( \begin{array}{cc}
            \widetilde{W} &  \widetilde{W}\\
            \mb{0}_{n\times n} & \widetilde{W}
      \end{array} \right),~ \mb{e}^{k+1}= \left( \begin{array}{cc}
            \mb{e}^{k+1}_x \\ \mb{e}^{k+1}_d
      \end{array} \right),
    \end{align} 
    along with $\mb{z}^{k+1}$, $\mb{x}^{k+1}$, and $\mb{e}_x^{k+1}$ following Eq.~\eqref{eq_loc_opt} in Lemma~\ref{lem_local_opt}. Further, the sequence   $\{\|\mb{x}^{k+1} - \mb{x}^* \|^2 + c^2 \|\mb{z}^{k+1} - \mb{y}^* \|^2 \}$ is convergent under dynamics \eqref{eq_y_sync}-\eqref{eq_x_sync}.
\end{lemma}
\begin{proof}
   To prove the convergence, from \cite[Appendix]{falsone2020tracking}, the following PD matrix is a valid Lyapunov matrix for the case of $\overline{\tau} = 0$,
   \begin{align} \label{eq_Q}
       Q = \left( \begin{array}{cc}
            2\mb{I}_n &  W_{inv} -2\mb{I}_n \\
            W_{inv} -2\mb{I}_n  & W_{inv}^{2}-2W_{inv} +2\mb{I}_n 
       \end{array} \right),
   \end{align}
with $W_{inv} = (\mb{I}_n-\widetilde{W})^{-1}$. 
For $\overline{\tau} \geq 1$, modify this PD $Q$ matrix by using $W_{inv} = (\mb{I}_n-\widetilde{W}^{\frac{1}{\overline{\tau}+1}})^{-1}$ in \eqref{eq_Q} instead, and the proof of Lyapunov stability under homogeneous delays $\overline{\tau}$ follows. This implies that,  

\small\begin{align} \nonumber
    \| \mb{1}_2 \otimes &c(\mb{z}^{k+1} - \mb{y}^*)- \mb{e}^{k+1}\|^2_Q = \| \mb{1}_2 \otimes c(\mb{z}^{k-\overline{\tau}} - \mb{y}^*)- \mb{e}^{k-\overline{\tau}}\|^2_Q \\ \label{eq_ineq_sych}
    &+2 (\mb{1}_2 \otimes c(\mb{z}^{k-\overline{\tau}} - \mb{y}^*))^\top Q (I_n -F)\mb{e}^{k-\overline{\tau}} -  \|  \mb{e}^{k-\overline{\tau}}\|^2_Q.
\end{align} \normalsize
Its combination with  Lemma~\ref{lem_local_opt} proves the convergence of the sequence $\{\|\mb{x}^{k+1} - \mb{x}^* \|^2 + c^2 \|\mb{z}^{k+1} - \mb{y}^* \|^2 \}$. 
\end{proof}


\begin{corollary} \label{cor_pw}
The \textit{stable} error dynamics \eqref{eq_F} in Lemma~\ref{lem_conv_sync} can be reshaped into its augmented form by replacing 
\begin{equation} \label{eq_pw_sync}
	\widetilde{PW} = \left( 
	\begin{array}{cccccc}
		\mb{0}_{n\times n} &  \mb{0}_{n\times n}  &  \hdots &  \mb{0}_{n\times n} & \widetilde{W}  \\
		\mb{I}_n &   \mb{0}_{n\times n}  &\hdots  & \mb{0}_{n\times n} & \mb{0}_{n\times n} \\
		\mb{0}_{n\times n} & \mb{I}_n &   \hdots  & \mb{0}_{n\times n} & \mb{0}_{n\times n}  \\
		\vdots & \vdots &  \ddots & \vdots & \vdots \\
		\mb{0}_{n\times n} & \mb{0}_{n\times n} &  \hdots & \mb{I}_n & \mb{0}_{n\times n}
	\end{array}
	\right)
\end{equation}
of size $\overline{\tau}+1$ into Eq. \eqref{eq_y_vect}-\eqref{eq_x_aug}. 
\end{corollary} 

\begin{lemma}  \label{lem_pwi}
    Define the \textit{modified} augmented matrix $\widetilde{PW}_r$ as follows: keep $P_r$ and put all other $P_{l} = \mb{0}_{n\times n}$, $l\neq r$, $l \in \{1,\dots,\overline{\tau}\}$ in  \eqref{eq_aug_WA}, i.e., the first row block as
    \small \begin{align} \label{eq_aug_WAi}
	\left[\widetilde{PW}_r\right]_{1:n,1:n(\overline{\tau}+1)} = \left( 
	\begin{array}{cccccc}
		\mb{0}_{n\times n}  &  \hdots & P_{r} \circ \widetilde{W}  & \hdots & \mb{0}_{n\times n} 
	\end{array}	
	\right).
\end{align} \normalsize
    Replace $\mb{I}_{n}$ blocks in $\widetilde{PW}$ by $\frac{1}{\overline{\tau}+1} \mb{I}_{n}$. Then, $\widetilde{PW} = \sum_{r=0}^{\overline{\tau}} \widetilde{PW}_r$ and $\rho(\widetilde{PW}_r) \leq (\frac{\rho(\widetilde{W})}{\overline{\tau}+1})^{\frac{1}{r+1}} < 1$ for $\rho(\widetilde{W}) < 1$.
\end{lemma}
\begin{proof}
    The proof follows similar to the proof of Lemma~5 in \cite{delay_est}. Note that the lower-triangular blocks of $\widetilde{PW}_r$ are scaled by $\frac{1}{\overline{\tau}+1}$.
    Let $\sigma_{P_r \circ \widetilde{W}}(\lambda)$ and $\sigma_{\widetilde{PW}_r}(\lambda)$ denote the characteristic equation of matrices $P_r \circ \widetilde{W}$ and $\widetilde{PW}_r$. Then, from Eq.~(25) in \cite[Appendix]{delay_est}, 
    \begin{align}
        \sigma_{\widetilde{PW}_r}(\lambda) \propto \lambda^{n(\overline{\tau}-r)} \sigma_{P_r \circ \widetilde{W}} \left( \left( (\overline{\tau}+1)\lambda\right)^{r+1}\right).
   \end{align}
    Following Remark~\ref{rem_augP} we have $\rho(P_r \circ \widetilde{W}) \leq \rho(\widetilde{W})$, and the proof follows.
\end{proof}

\begin{lemma}[Heterogeneous Delays] \label{lem_asynch}
    Let Assumptions \ref{ass_net} and \ref{ass_delay} hold. Then, $\rho(\widetilde{PW}) \leq \rho(\widetilde{W})^{\frac{1}{\overline{\tau}+1}} < 1$ for $\rho(\widetilde{W}) < 1$.
\end{lemma} 
\begin{proof}
    The proof follows from Lemma~\ref{lem_synch} and \cite[Lemma~5]{delay_est}. Recall that, for  heterogeneous delays $\tau_{ij} = r \leq \overline{\tau}$, function $\rho(\widetilde{W})^{\frac{1}{r+1}}$ is an increasing function of $r$ for $\rho(\widetilde{W}) < 1$, i.e., $\rho(\widetilde{W})^{\frac{1}{r+1}} \leq \rho(\widetilde{W})^{\frac{1}{\overline{\tau}+1}} < 1$.
    Therefore, from the definition of $\widetilde{PW}$ in \eqref{eq_aug_WA}, the proof follows for  general heterogeneous delays $\tau_{ij} = r < \overline{\tau}$.
\end{proof}

Using Corollary~\ref{cor_xyz}-\ref{cor_pw} and Lemma~\ref{lem_pwi}-\ref{lem_asynch} one can relate the convergence and optimality of the homogeneous and heterogeneous cases. Then, \textbf{Proof of Theorem \ref{thm_conv}} is as follows.  

\begin{proof}
The proof follows similar to Lemma~\ref{lem_conv_sync}. First, rewrite the error dynamics~\eqref{eq_error_compact} as,
   \begin{align} \label{eq_F0}
        \underline{\mb{e}}^{k+1} = \underline{F} \underline{\mb{e}}^{k} +  \mb{1}_2 \otimes c(\underline{\mb{z}}^{k+1} - \underline{\mb{z}}^{k})
    \end{align}
    with
   \begin{align}
      \underline{F} = \left( \begin{array}{cc}
            \widetilde{PW} &  \widetilde{PW}\\
            \mb{0}_{n} & \widetilde{PW}
      \end{array} \right),~ \underline{\mb{e}}^{k} = \left( \begin{array}{cc}
            \underline{\mb{e}}_k^x \\ \underline{\mb{e}}_k^d
      \end{array} \right).
    \end{align} 
   From Lemma~\ref{lem_asynch},   $\rho(\widetilde{W})^{\frac{1}{\overline{\tau}+1}} < 1$ and the eigenvalues lie inside the unit circle; thus, $\rho(\underline{F}) < 1$ which implies stable error dynamics in the absence of $\underline{\mb{z}}$ terms. Following the input-to-state stability as in Lemma~\ref{lem_conv_sync} in the presence of $\underline{\mb{z}}$ terms, rewrite $\underline{F} =  \sum_{r=0}^{\overline{\tau}} \underline{F}_r$ with 
   \begin{align}
      \underline{F}_r = \left( \begin{array}{cc}
            \widetilde{PW}_r &  \widetilde{PW}_r\\
            \mb{0}_{n} & \widetilde{PW}_r
      \end{array} \right)
   \end{align} 
   and $\widetilde{PW}_r$ defined as in Lemma~\ref{lem_pwi}. Then, $\underline{F}_r$s follow similar structure as in Corollary~\ref{cor_pw} and Lemma~\ref{lem_pwi} given for the homogeneous case. Then, the proof of Lyapunov stability and convergence can be restated similar to the homogeneous case for $r=0,\dots,\overline{\tau}$ and proof of part 1) follows similar to Lemma~\ref{lem_conv_sync}. Recall that part 1) implies that as $k\rightarrow \infty$ the augmented vectors $\underline{\mb{d}}^k$, $\underline{\mb{x}}^k$, and $\underline{\mb{y}}^k$ respectively converge to $\mb{1}_{\overline{\tau}+1} \otimes \mb{d}^k$, $\mb{1}_{n(\overline{\tau}+1)} \otimes \overline{x}^k$, and $\mb{1}_{\overline{\tau}+1} \otimes \mb{y}^k$. This follows the definition of error variables and $\overline{d}^k$ and implies that Corollary~\ref{cor_xyz} holds. Then, in limit as $k\rightarrow \infty$, part 2) and 3) can be restated for non-augmented vector variables $\mb{x}^k$, $\mb{x}^*$, $\mb{y}^*$ and $\mb{z}^k$. The proof of part 2) follows from Corollary~\ref{cor_xyz}. 
   The proof of part 3) follows from \cite[Lemma~4.1]{bertsekas1989parallel} for the \textsf{Parallel-ADMM} solution \eqref{eq_paral_y}-\eqref{eq_paral_x}, \cite[Theorem~2]{falsone2020tracking} for the non-delayed \textsf{Distributed-Parallel-ADMM} \eqref{eq_nodelay_y}-\eqref{eq_nodelay_x}, and Lemma~\ref{lem_local_opt}-\ref{lem_conv_sync} for the homogeneous case. The sketch of the proof follows as: \emph{(i)} adding and subtracting $(\underline{\mb{y}}^{k+1} - \underline{\mb{y}}^*)^\top \overline{\mb{x}}_{\overline{\tau}}^{k+1}$ to the augmented version of \eqref{eq_loc_opt_proof}, \emph{(ii)} $(\underline{\mb{y}}^{k+1} - \underline{\mb{y}}^*)^\top \overline{\mb{x}}_{\overline{\tau}}^{k+1} =  (\underline{\mb{y}}^{k+1} - \mb{b})^\top \overline{\mb{x}}_{\overline{\tau}}^{k+1} = n(\overline{\tau}+1) \overline{x}^{k+1} \overline{d}^{k+1}$ which results from $\mb{1}_n^\top\mb{y}^*=b$, \emph{(iii)} the boundedness of $\|\underline{\mb{e}}_d^k\|$ and $\overline{d}^k$ implies bounded $\|\underline{\mb{y}}^k\|$, $\|\mb{z}^k\|$ from part 2), \emph{(iv)} $\lim_{k\rightarrow \infty} \overline{x}^k \overline{d}^k = 0$ and $\lim_{k\rightarrow \infty} (\underline{\mb{y}}^{k+1} - \underline{\mb{y}}^*)^\top \underline{e}^{k+1}_x = 0$, \emph{(v)} $\lim_{k\rightarrow \infty} \sup \Phi(\mb{y}) \leq \Phi(\mb{y}^*)$ and $\Phi(\mb{y}^{k}) \geq \Phi(\mb{y}^*)$ for any feasible $\mb{y}$. Combining the steps \emph{(iv)}-\emph{(v)} and for sufficiently large $k$ to reach feasibility, $\lim_{k\rightarrow \infty} \mb{y}^k = \mb{y}^*$. The proof for local dual variables $x_i^k$ follows similarly: \emph{(i)} in limit as $k\rightarrow \infty$ Corollary~\ref{cor_xyz} holds, and from the definition of $\underline{\pmb{\delta}}^k$ and $\underline{\pmb{\eta}}^k$ and some algebraic manipulations, Eq.~\eqref{eq_bertsek_2}-\eqref{eq_bertsek} can be rewritten as $\phi_i(y_i^{k+1}) +  (y_i^{k+1} - b_i) x_i^{k+1} \leq \phi_i(y_i) +  (y_i - b_i) x_i^{k+1}$ for all $y_i \in \Y_i$, \emph{(ii)} this, by strong duality in Assumption~\ref{ass_dual}, implies that $f_i(x_i^{k+1}) = \phi_i(y_i^{k+1}) + x_i^{k+1}(y_i^{k+1}-b_i)$ and having $\lim_{k\rightarrow \infty} \sum_{i=1}^n f_i(x_i^{k}) = \Phi(\mb{y}^*) \geq \max_{x\in \mathbb{R}} \sum_{i=1}^n f_i(x)$ and $\lim_{k\rightarrow \infty} \sum_{i=1}^n f_i(x_i^{k}) \leq \phi_i(y_i) + (y_i-b_i) \widetilde{x}$ with  $\widetilde{x} := \lim_{k\rightarrow 0} x_i^{k} $ from part 1), and \emph{(iii)} from part 2), $\lim_{k \rightarrow \infty} \{\| \overline{\mb{x}} - \mb{x}^* \|^2 + c^2 \|\mb{z}^{k} - \mb{y}^* \|^2 \} = 0$. Thus,
   $$\lim_{k\rightarrow \infty} \underline{\mb{x}}^k = \mb{1}_{\overline{\tau}+1} \overline{\mb{x}} = \mb{1}_{n(\overline{\tau}+1)} x^* = \underline{\mb{x}}^*,$$
   and the proof follows.
\end{proof}  
   
\bibliographystyle{IEEEtran}
\bibliography{bibliography}

\begin{thebibliography}{10}
\providecommand{\url}[1]{#1}
\csname url@samestyle\endcsname
\providecommand{\newblock}{\relax}
\providecommand{\bibinfo}[2]{#2}
\providecommand{\BIBentrySTDinterwordspacing}{\spaceskip=0pt\relax}
\providecommand{\BIBentryALTinterwordstretchfactor}{4}
\providecommand{\BIBentryALTinterwordspacing}{\spaceskip=\fontdimen2\font plus
\BIBentryALTinterwordstretchfactor\fontdimen3\font minus
  \fontdimen4\font\relax}
\providecommand{\BIBforeignlanguage}[2]{{%
\expandafter\ifx\csname l@#1\endcsname\relax
\typeout{** WARNING: IEEEtran.bst: No hyphenation pattern has been}%
\typeout{** loaded for the language `#1'. Using the pattern for}%
\typeout{** the default language instead.}%
\else
\language=\csname l@#1\endcsname
\fi
#2}}
\providecommand{\BIBdecl}{\relax}
\BIBdecl

\bibitem{Themis_allerton}
C.~N. Hadjicostis and T.~Charalambous, ``Asynchronous coordination of
  distributed energy resources for the provisioning of ancillary services,'' in
  \emph{49th Allerton Conference on Communication, Control, and Computing},
  2011, pp. 1500--1507.

\bibitem{mrd20211st}
M.~Doostmohammadian, A.~Aghasi, M.~Vrakopoulou, and T.~Charalambous,
  ``1st-order dynamics on nonlinear agents for resource allocation over
  uniformly-connected networks,'' in \emph{IEEE Conference on Control
  Technology and Applications}, 2022, arXiv preprint arXiv:2109.04822.

\bibitem{rikos2021optimal}
A.~I. Rikos, A.~Grammenos, E.~Kalyvianaki, C.~N. Hadjicostis, T.~Charalambous,
  and K.~H. Johansson, ``Optimal cpu scheduling in data centers via a
  finite-time distributed quantized coordination mechanism,'' in \emph{51st
  IEEE Conf. on Decision and Control}, 2021.

\bibitem{shames2011accelerated}
E.~Ghadimi, M.~Johansson, and I.~Shames, ``Accelerated gradient methods for
  networked optimization,'' in \emph{American Control Conference}.\hskip 1em
  plus 0.5em minus 0.4em\relax IEEE, 2011, pp. 1668--1673.

\bibitem{boyd2006optimal}
L.~Xiao and S.~Boyd, ``Optimal scaling of a gradient method for distributed
  resource allocation,'' \emph{Journal of Optimization Theory and
  Applications}, vol. 129, no.~3, pp. 469--488, 2006.

\bibitem{cherukuri2015distributed}
A.~Cherukuri and J.~Cort{\'e}s, ``Distributed generator coordination for
  initialization and anytime optimization in economic dispatch,'' \emph{IEEE
  Transactions on Control of Network Systems}, vol.~2, no.~3, pp. 226--237,
  2015.

\bibitem{aybat2016distributed}
N.~S. Aybat and E.~Y. Hamedani, ``Distributed primal-dual method for
  multi-agent sharing problem with conic constraints,'' in \emph{50th Asilomar
  Conference on Signals, Systems and Computers}.\hskip 1em plus 0.5em minus
  0.4em\relax IEEE, 2016, pp. 777--782.

\bibitem{Nedic2018}
A.~Nedi{\'c}, A.~Olshevsky, and W.~Shi, ``Improved convergence rates for
  distributed resource allocation,'' in \emph{IEEE Conference on Decision and
  Control}, 2018, pp. 172--177.

\bibitem{falsone2020tracking}
A.~Falsone, I.~Notarnicola, G.~Notarstefano, and M.~Prandini, ``Tracking-{ADMM}
  for distributed constraint-coupled optimization,'' \emph{Automatica}, vol.
  117, p. 108962, 2020.

\bibitem{fast}
M.~Doostmohammadian, A.~Aghasi, M.~Pirani, E.~Nekouei, U.~A. Khan, and
  T.~Charalambous, ``Fast-convergent anytime-feasible dynamics for distributed
  allocation of resources over switching sparse networks with quantized
  communication links,'' in \emph{European Control Conference}, 2022, arXiv
  preprint arXiv:2012.08181.

\bibitem{wu2020distributed}
J.~Wu, Q.~Deng, T.~Han, and H.~Yan, ``Distributed bipartite tracking consensus
  of nonlinear multi-agent systems with quantized communication,''
  \emph{Neurocomputing}, vol. 395, pp. 78--85, 2020.

\bibitem{garg_cdc}
K.~Garg, M.~Baranwal, and D.~Panagou, ``A fixed-time convergent distributed
  algorithm for strongly convex functions in a time-varying network,'' in
  \emph{IEEE Conference on Decision and Control}, 2020, pp. 4405--4410.

\bibitem{taes2020finite}
M.~Doostmohammadian, ``Single-bit consensus with finite-time convergence:
  Theory and applications,'' \emph{IEEE Transactions on Aerospace and
  Electronic Systems}, vol.~56, no.~4, pp. 3332--3338, 2020.

\bibitem{wu2021new}
X.~Wu, S.~Magnusson, and M.~Johansson, ``A new family of feasible methods for
  distributed resource allocation,'' in \emph{IEEE Conference on Decision and
  Control}, 2021, pp. 3355--3360.

\bibitem{banjac2019decentralized}
G.~Banjac, F.~Rey, P.~Goulart, and J.~Lygeros, ``Decentralized resource
  allocation via dual consensus {ADMM},'' in \emph{American Control
  Conference}.\hskip 1em plus 0.5em minus 0.4em\relax IEEE, 2019, pp.
  2789--2794.

\bibitem{carli_admm}
R.~Carli and M.~Dotoli, ``Distributed alternating direction method of
  multipliers for linearly constrained optimization over a network,''
  \emph{IEEE Control Systems Letters}, vol.~4, no.~1, pp. 247--252, 2020.

\bibitem{chang_admm}
T.~Chang, M.~Hong, and X.~Wang, ``Multi-agent distributed optimization via
  inexact consensus {ADMM},'' \emph{IEEE Transactions on Signal Processing},
  vol.~63, no.~2, pp. 482--497, 2015.

\bibitem{arxiv_digraph}
\BIBentryALTinterwordspacing
K.~Rokade and R.~K. Kalaimani, ``Distributed {ADMM} over directed networks,''
  2021. [Online]. Available: \url{https://arxiv.org/abs/2010.10421}
\BIBentrySTDinterwordspacing

\bibitem{Teixeira_admm}
A.~Teixeira, E.~Ghadimi, I.~Shames, H.~Sandberg, and M.~Johansson, ``The {ADMM}
  algorithm for distributed quadratic problems: Parameter selection and
  constraint preconditioning,'' \emph{IEEE Transactions on Signal Processing},
  vol.~64, no.~2, pp. 290--305, 2016.

\bibitem{Wei_cdc}
W.~Jiang, A.~Grammenos, E.~Kalyvianaki, and T.~Charalambous, ``An asynchronous
  approximate distributed alternating direction method of multipliers in
  digraphs,'' in \emph{IEEE Conf. on Decision and Control}, 2021.

\bibitem{wang2019distributed}
X.~Wang, Y.~Hong, X.~Sun, and K.~Liu, ``Distributed optimization for resource
  allocation problems under large delays,'' \emph{IEEE Transactions on
  Industrial Electronics}, vol.~66, no.~12, pp. 9448--9457, 2019.

\bibitem{boyd2004convex}
S.~P. Boyd and L.~Vandenberghe, \emph{Convex optimization}.\hskip 1em plus
  0.5em minus 0.4em\relax Cambridge university press, 2004.

\bibitem{Themis_delay}
C.~N. Hadjicostis and T.~Charalambous, ``Average consensus in the presence of
  delays in directed graph topologies,'' \emph{IEEE Transactions on Automatic
  Control}, vol.~59, no.~3, pp. 763--768, 2013.

\bibitem{bertsekas1989parallel}
D.~Bertsekas and J.~Tsitsiklis, \emph{Parallel and distributed computation:
  numerical methods}.\hskip 1em plus 0.5em minus 0.4em\relax Athena Scientific,
  1989.

\bibitem{delay_est}
M.~Doostmohammadian, M.~Pirani, U.~A. Khan, and T.~Charalambous,
  ``Consensus-based distributed estimation in the presence of heterogeneous,
  time-invariant delays,'' \emph{IEEE Control Systems Letters}, vol.~6, pp.
  1598 -- 1603, 2021.

\bibitem{995554}
R.~A. Berry and R.~G. Gallager, ``Communication over fading channels with delay
  constraints,'' \emph{IEEE Transactions on Information Theory}, vol.~48,
  no.~5, pp. 1135--1149, 2002.

\bibitem{vrakopoulou2017chance}
M.~Vrakopoulou, B.~Li, and J.~L. Mathieu, ``Chance constrained reserve
  scheduling using uncertain controllable loads part i: Formulation and
  scenario-based analysis,'' \emph{IEEE Transactions on Smart Grid}, vol.~10,
  no.~2, pp. 1608--1617, 2017.

\bibitem{mrd_vtc2022}
M.~Doostmohammadian, M.~Vrakopoulou, A.~Aghasi, and T.~Charalambous,
  ``Distributed finite-sum constrained optimization subject to nonlinearity on
  the node dynamics,'' in \emph{IEEE Vehicular Technology Conference}, 2022,
  arXiv preprint arxiv.org/abs/2203.14527.

\bibitem{magnusson2018communication}
S.~Magn{\'u}sson, C.~Enyioha, N.~Li, C.~Fischione, and V.~Tarokh,
  ``Communication complexity of dual decomposition methods for distributed
  resource allocation optimization,'' \emph{IEEE Journal of Selected Topics in
  Signal Processing}, vol.~12, no.~4, pp. 717--732, 2018.

\end{thebibliography}

\end{document}